\documentclass{vldb}
\usepackage{graphicx}
\usepackage{caption}
\usepackage{subcaption}
\usepackage{url}
\usepackage{footnote}
\usepackage{appendix}
\usepackage{pdfpages}
\usepackage{balance}  
\usepackage{ifpdf}
\usepackage{amsmath}
\usepackage{tabularx}
\usepackage{booktabs}
\usepackage{hhline}
\usepackage{array}
\usepackage{multirow}
\usepackage[ruled, vlined, boxed, linesnumbered]{algorithm2e}
\usepackage[none]{hyphenat}
\usepackage{color, colortbl}

\usepackage{tikz}
\usepackage{tkz-euclide}
\usetikzlibrary{shapes,arrows}
\usepackage{pgfplots}
\usepackage{pgfplotstable}

\usepackage{array}
\newcolumntype{C}[1]{>{\centering\arraybackslash}m{#1}}

\usepackage{qtree}
\usepackage{xcolor}
\usetikzlibrary{positioning}
\tikzset{main node/.style={circle,draw,minimum size=0.3cm,inner sep=0pt},  }

\definecolor{Gray}{gray}{0.5}

\newtheorem{thm}{Theorem}
\newtheorem{lem}[thm]{Lemma}

\newtheorem{defn}{Definition}

\def\id#1{\ensuremath{\mathit{#1}}}

\def\idrm#1{\ensuremath{\mathrm{#1}}}

\begin{document}

\title{MSQ-Index: A Succinct Index for Fast Graph Similarity Search}

\numberofauthors{4}

\author{
\alignauthor
Xiaoyang Chen \\
       \affaddr{Xidian University}\\
       \affaddr{Xi'an 710071, China}\\
       \email{chenxyu1991@gmail.com}
\alignauthor
Hongwei Huo\\
       \affaddr{Xidian University}\\
       \affaddr{Xi'an 710071, China}\\
       \email{hwhuo@mail.xidian.edu.cn}
\alignauthor
Jun Huan \\
       \affaddr{The University of Kansas}\\
       \affaddr{Lawrence, KS 66045, USA}\\
       \email{jhuan@ku.edu}
\and  
\alignauthor
Jeffrey Scott Vitter\\
       \affaddr{The University of Mississippi}\\
       \affaddr{Oxford, MS 38677-1848, USA}\\
       \email{JSV@OleMiss.edu}
}

\maketitle
\begin{abstract}

Graph similarity search has received considerable attention in
many applications, such as bioinformatics, data mining, pattern
recognition, and social networks. Existing methods for this problem
have limited scalability because of the huge amount of memory
they consume when handling very large graph databases with millions
or billions of graphs.

In this paper, we study the problem of graph similarity search
under the graph edit distance constraint. We present a space-efficient
index structure based upon the $q$-$\id{gram}$ tree that incorporates
succinct data structures and hybrid encoding to achieve improved
query time performance with minimal space usage. Specifically,
the space usage of our index requires only 5\%--15\% of the previous
state-of-the-art indexing size on the tested data while at the
same time achieving 2--3 times acceleration in query time with small
data sets. We also boost the query performance by augmenting the
global filter with range search, which allows us to perform a query
in a reduced region. In addition, we propose two effective filters
that combine degree structures and label structures. Extensive
experiments demonstrate that our proposed approach is superior
in space and competitive in filtering to the state-of-the-art
approaches. To the best of our knowledge, our index is the first
in-memory index for this problem that successfully scales to cope
with the large dataset of~25 million chemical structure graphs
from the PubChem dataset.

\end{abstract}

\section{Introduction}

Graphs are widely used to model complicated data objects in many
disciplines, such as bioinformatics, social networks, software
and data engineering. Effective analysis and management of graph
data become increasingly important. Many queries have been investigated
and they can be roughly divided into two broad categories: graph exact
search~\cite{YanYH2004} and graph similarity search~\cite{ShangZLYR2010}.
Compared with exact search, similarity search can provide a robust
solution that permits error-tolerant and supports to search patterns
that are not precisely defined.

Similarity computation between two attributed graphs is a core
operation of graph similarity search and it has been used in
various applications such as pattern recognition, graph classification
and chemistry analysis~\cite{KellyH2005}. There are at least four
metrics being well investigated: graph edit distance~\cite{GoudaH2016,WangWYY2012,ZengTWFZ2009,ZhaoXLLZ2013,ZhengZLWZ2015},
maximal common subgraph distance~\cite{FernandezV2001}, graph
alignment~\cite{FrohlichWS2005} and graph kernel functions~\cite{ShervashidzeB2009,WangSHG2009}.
In this paper, we focus on the graph edit distance since it is
applicable to virtually all types of data graphs and can also
capture precisely structural differences. The graph edit distance~$\id{ged}(g, h)$
between two graphs~$g$ and~$h$ is defined as the minimum number
of edit operations needed to transform one graph to another.

Given a graph database~$G$, a query graph~$h$ and an edit distance
threshold~$\tau$, the graph similarity search problem aims to find
all graphs~$g$ in~$G$ satisfying $ged(g, h)\leq~\tau$. Unfortunately,
computing the graph edit distance is known to be an NP-hard
problem~\cite{ZengTWFZ2009}. Therefore, for a large transaction
database, such as PubChem, which stores information about roughly~50
million chemical compounds, similarity search is very challenging.

Most of the existing methods adopt the filter-and-verify schema
to speed up the search. With such a schema, we first filter data
graphs that are not possible results to generate a candidate set,
and then validate the candidate graphs with the expensive graph edit
distance computations. In general, the existing filters can be divided
into four categories: global filter, $q$-$\id{gram}$ counting filter,
mapping distance-based filter and disjoint partition-based filter.
Specifically, number count filter~\cite{ZengTWFZ2009} and label
count filter~\cite{ZhaoXLW2012} are two global filters. The former
is derived based upon the differences of the number of vertices and
edges of comparing graphs. The later takes labels as well as structures
into account, further improving the former. $\kappa$-AT~\cite{WangWYY2012}
and GSimJoin~\cite{ZhaoXLW2012} are two major~$q$-$\id{gram}$ counting
filters. They considered a $\kappa$-$\id{adjacent}$ subtree and
a simple path of length $p$ as a $q$-$\id{gram}$, respectively.
C-Star~\cite{ZengTWFZ2009} and Mixed~\cite{ZhengZLWZ2013,ZhengZLWZ2015}
are two major mapping distance-based filters. The lower bounds are
derived based on the minimum weighted bipartite graphs between the
star and branch structures of comparing graphs, respectively.
Pars~\cite{ZhaoXLLZ2013} is a disjoint partition-based filter.
It divides each data graph~$g$ into several disjoint substructures
and prunes~$g$ by the subgraph isomorphism.

Even though promising preliminary results have been achieved by
existing methods GSimJoin~\cite{ZhaoXLW2012}, C-Star ~\cite{ZengTWFZ2009}
and Mixed ~\cite{ZhengZLWZ2015}, our empirical evaluation of all
the methods aforementioned showed that they are not scalable to
large graph databases. The critical limitation of existing methods
are: (1) existing filters having a weak filter ability produce large
candidate sets, resulting in an unacceptable computational cost
for verification, (2) the index storage cost of the existing methods
is too expensive to run properly. For example, for a database of~10
million graphs, C-Star on average produces $5 \times 10^5$ number
of candidates for verification when $\tau = 5$. Both GSimJoin and
Mixed produce an index that is too large to fit into the main memory
for large input data. The details of the empirical study are
presented in Section~\ref{sec:experiments}.

To solve the above issues, we propose a space-efficient index structure
for graph similarity search which significantly reduces the storage
space. Our contributions in this paper are summarized below.

\begin{itemize}

\item
  We propose two effective filters, i.e. degree-based~$q$-$\id{gram}$
  counting filter and degree-sequence filter, by using the degree
  structures and label structures.
\item
  We create a $q$-$\id{gram}$ tree to speed up filtering process.
  More importantly, we propose the succinct representation of the
  $q$-$\id{gram}$ tree which combines with hybrid coding, significantly
  reducing the space required for the representation of the $q$-$\id{gram}$
  tree.
\item
  We convert the number count filter to a two-dimensional orthogonal
  range searching, which helps us perform a query at a reduced region
  and hence further improves the filtering performance.
\item
  We have conducted extensive experiments over both real and synthetic
  datasets to evaluate the index storage space, construction time,
  filtering capability, and response time. The result is graph similarity
  search index that we refer to as ``MSQ-Index''. It confirms the effectiveness
  and efficiency of our proposed approaches and show that our method
  can scale well to cope with the large dataset of 25 million chemical
  compounds from the PubChem dataset.

\end{itemize}

The rest of this paper is organized as follows: In Section~\ref{sec:preliminaries},
we introduce the problem definition. In Section~\ref{sec:filters},
we present the degree-based $q$-$\id{gram}$ counting filter and the
degree-sequence filter. In Section~\ref{sec:queryRegion}, we give
a method to reduce the query region. In Section~\ref{sec:Q-gramIndex},
we introduce the index structure. In Section~\ref{sec:queryProcessing},
we give the query algorithm. In Section~\ref{sec:experiments}, we report
the experimental results. We investigate the research work related to
this paper in Section~\ref{sec:relatedWorks}. Finally, we make concluding
remarks in Section~\ref{sec:conclusion}.

\section{Preliminaries}
\label{sec:preliminaries}

In this section, we introduce the basic notations and definitions
of graph edit distance and graph similarity search.

\begin{defn}[Attributed Graph]
\label{def:graph}

A labeled graph is defined as a six-tuple $g = (V_g, E_g, \mu,\zeta,\Sigma_{V_g}, \Sigma_{E_g})$,
where~$V_g$ is the set of vertices, $E_g \subseteq V_g \times V_g$
is the set of edges, $\mu:V_g \rightarrow \Sigma_{V_g}$ is the
vertex labeling function which assigns a label~$\mu(v)$ to the
vertex~$v$, $\zeta:E_g \rightarrow \Sigma_{E_g}$ is the edge
labeling function which assigns a label~$\zeta(e)$ to the
edge~$e$, $\Sigma_{V_g}$ and $\Sigma_{E_g}$ are the label
multisets of~$V_g$ and~$E_g$, respectively.

\end{defn}

In this paper, we only focus on simple undirected graphs
without multi-edge or self-loop. We use $|V_g|$ and $|E_g|$
to denote the number of vertices and edges in $g$, respectively.
The graph size refers to $|V_g|$ in this paper. Although in the
following discussion we only focus on undirected graphs, our
methods can be extended to handle directed graphs.

\begin{defn}[Graph Isomorphism {\cite{YanYH2004}}]
\label{def:isomorphism}

Given two graphs $g$ and $h$, an isomorphism of graphs $g$ and
$h$ is a bijection $f: V_g \rightarrow V_h$, such that (1) for
all $v\in V_g$, $f(v) \in V_h$ and $\mu(v) = \mu(f(v))$. (2) for
all~$e(u, v) \in E_g$, $e(f(u), f(v)) \in E_h$ and $\zeta(e(u, v)) = \zeta(e(f(u), f(v)))$.
If~$g$ is isomorphic to~$h$, we denote~$g \cong h$.

\end{defn}

There are six primitive edit operations that can transform one
graph to another~\cite{JusticeH2006}. These edit operations are
inserting/deleting an isolated vertex, inserting/deleting an edge
between two vertices and substituting the label of a vertex or
an edge. We denote the substitution of two vertices~$u$ and
$v$ by~$(u \rightarrow v)$, the deletion of vertex~$u$
by~$(u \rightarrow \epsilon)$, and the insertion of vertex~$v$
by $(\epsilon \rightarrow v)$. For edges, we use a similar notation.
Given two graphs~$g$ and $h$, an edit path $P = \langle p_1, p_2, \ldots , p_k\rangle$
is a sequence of edit operations that transforms $h$ to $g$,
such as $h = h^0 \xrightarrow{p_1} h^1\xrightarrow{p_2} \ldots \xrightarrow{p_k} h^k \cong g$.
In Figure~\ref{fig:Fig01}, we give an example of an edit path~$P$
between~$g$ and~$h$, where the vertex labels are represented
by different symbols. The length of $P$ is~6, which consists
of two edge deletions, one vertex deletion, one vertex insertion
and two edge insertions. In the following sections, we use~$|P|$
to denote the length of $P$.

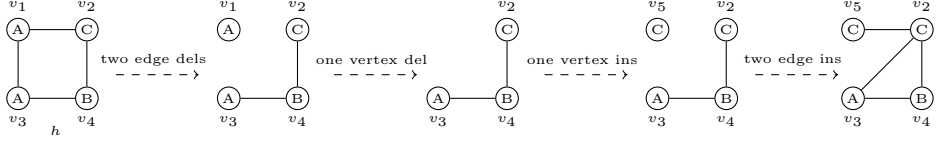
\begin{figure*}[!htbp]
\centering
\begin{tikzpicture}[xshift=-1cm]
\tikzset{main node/.style={circle,draw,minimum size=0.3cm,inner sep=0pt},  }
\tiny{

    \node[main node] (1) [label=above:$v_1$] {A};
    \node[main node] (2) [right = 0.6cm  of 1,label=above:$v_2$] {C};
    \node[main node] (3) [below = 0.6cm  of 1,label=below:$v_3$] {A};
    \node[main node] (4) [right = 0.6cm  of 3,label=below:$v_4$] {B};

    \path[draw,thin]
    (1) edge node {} (2)
    (1) edge node {} (3)
    (2) edge node {} (4)
    (3) edge node {} (4)
    ;

    \draw[->,dashed]  (1.3,-0.6) -- (2.4,-0.6);
    \node at (1.8, -0.4) {\tiny{two edge dels}};

    \begin{scope}[xshift=2.8cm]
    \node[main node] (5) [label=above:$v_1$] {A};
    \node[main node] (6) [right = 0.6cm  of 5,label=above:$v_2$] {C};
    \node[main node] (7) [below = 0.6cm  of 5,label=below:$v_3$] {A};
    \node[main node] (8) [right = 0.6cm  of 7,label=below:$v_4$] {B};

    \path[draw,thin]
    (6) edge node {} (8)
    (7) edge node {} (8)
    ;
    \end{scope}

    \draw[->,dashed]  (4.2,-0.6) -- (5.3,-0.6);
    \node at (4.7, -0.4) {\tiny{one vertex del}};

    \begin{scope}[xshift=6.5cm]
    \node[main node] (1) [label=above:$v_2$] {C};
    \node[main node] (2) [below = 0.6cm  of 1,label=below:$v_4$] {B};
    \node[main node] (3) [left = 0.6cm  of 2,label=below:$v_3$] {A};

    \path[draw,thin]
    (1) edge node {} (2)
    (3) edge node {} (2)
    ;
    \end{scope}

    \draw[->,dashed]  (7.0,-0.6) -- (8.1,-0.6);
    \node at (7.5, -0.4) {\tiny{one vertex ins}};

    \begin{scope}[xshift=8.5cm]
    \node[main node] (1) [label=above:$v_5$] {C};
    \node[main node] (2) [right = 0.6cm  of 1,label=above:$v_2$] {C};
    \node[main node] (3) [below = 0.6cm  of 1,label=below:$v_3$] {A};
    \node[main node] (4) [right = 0.6cm  of 3,label=below:$v_4$] {B};

    \path[draw,thin]
    (2) edge node {} (4)
    (3) edge node {} (4)
    ;
    \end{scope}

   \draw[->,dashed]  (9.8,-0.6) -- (10.9,-0.6);
    \node at (10.3, -0.4) {\tiny{two edge ins}};

    \begin{scope}[xshift=11.1cm]
    \node[main node] (1) [label=above:$v_5$] {C};
    \node[main node] (2) [right = 0.6cm  of 1,label=above:$v_2$] {C};
    \node[main node] (3) [below = 0.6cm  of 1,label=below:$v_3$] {A};
    \node[main node] (4) [right = 0.6cm  of 3,label=below:$v_4$] {B};

    \path[draw,thin]
    (1) edge node {} (2)
    (2) edge node {} (3)
    (2) edge node {} (4)
    (3) edge node {} (4)
    ;
    \end{scope}

}

\tkzText[below](0.5,-1.2){$h$}
\end{tikzpicture}
\caption{\small{An edit path $P$ between graphs $g$ and $h$.}}
\label{fig:Fig01}
\end{figure*}

\begin{defn}[Optimal Edit Path]
\label{def:editPath}

 Given two \\ graphs~$g$ and $h$, an edit path $P$ between~$g$
 and~$h$  is an optimal edit path if and only if there does not
 exist another  edit path~$P'$ such that $|P'| < |P|$. The graph
 edit distance between them, denoted by $\id{ged}(g, h)$, is the
 length of the optimal edit path.

\end{defn}

\noindent \textbf{Problem statement}: Given a graph database $G = \{g_1, g_2,
\\ {\dots}, \id{g_{|G|}}\}$, a query graph $h$, and an edit distance
$\idrm{threshold}\ \tau$, the problem is to find all the graphs~$g$
in~$G$ such that $\id{ged}(g, h) \leq~\tau$, where~$\id{ged}(g, h)$
is the graph edit distance of graphs~$g$ and~$h$ defined in Definition~\ref{def:editPath}.

Figure~\ref{fig:Fig02} shows a query graph $h$ and three data
$\idrm{graphs}$~$g_1, g_2$, and $g_3$. We can obtain that $\id{ged}(g_1, h) = 3$,
$\id{ged}(g_2, h) = 4$, and~$\id{ged}(g_3, h) = 3$. If the edit
distance threshold~$\tau$ = 3, $g_1$ and $g_3$ are the required graphs.

\begin{figure}[!htbp]
\centering
\begin{tikzpicture}[xshift=0.4cm]
\tiny{
    \node[main node] (1) {A};
    \node[main node] (2) [right = 0.6cm  of 1] {A};
    \node[main node] (3) [below = 0.6cm  of 1] {C};
    \node[main node] (4) [right = 0.6cm  of 3] {B};

    \path[draw,thin]
    (1) edge node {} (2)
    (1) edge node {} (3)
    (2) edge node {} (4)
    (3) edge node {} (4)
    ;

    \begin{scope}[xshift=1.6cm]
    \node[main node] (1) {C};
    \node[main node] (2) [below = 0.6cm  of 1] {A};
    \node[main node] (3) [right = 0.6cm  of 2] {A};

    \path[draw,thin]
    (1) edge node {} (2)
    (2) edge node {} (3)
    ;
    \end{scope}

    \begin{scope}[xshift=4.0cm]
    \node[main node] (1) {A};
    \node[main node] (2) [below = 0.6cm  of 1] {C};
    \node[main node] (3) [left = 0.6cm  of 2] {A};
    \node[main node] (4) [right = 0.6cm  of 2] {A};

    \path[draw,thin]
    (1) edge node {} (2)
    (2) edge node {} (3)
    (2) edge node {} (4)
    ;
    \end{scope}

    \begin{scope}[xshift=5.5cm]
    \node[main node] (1) {C};
    \node[main node] (2) [below = 0.6cm  of 1] {A};
    \node[main node] (3) [right = 0.6cm  of 1] {C};
    \node[main node] (4) [right = 0.6cm  of 2] {B};

    \path[draw,thin]
    (1) edge node {} (2)
    (2) edge node {} (3)
    (2) edge node {} (4)
    (3) edge node {} (4)
    ;
    \end{scope}

\tkzText[below](0.5,-1.1){$h$}
\tkzText[below](2.1,-1.1){$g_1$}
\tkzText[below](4.0,-1.1){$g_2$}
\tkzText[below](6.0,-1.1){$g_3$}
}
\end{tikzpicture}
\caption{\small{Query graph $h$ and data graphs $g_1$, $g_2$, and $g_3$.}}
\label{fig:Fig02}
\end{figure}
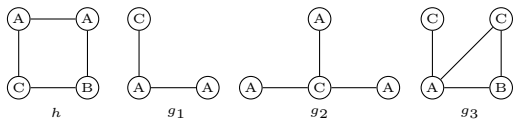

The computation of graph edit distance is an NP-hard problem~\cite{ZengTWFZ2009}.
The state-of-the-art approaches like~\cite{WangWYY2012,ZengTWFZ2009,ZhaoXLLZ2013,ZhaoXLW2012,ZhengZLWZ2013}
for graph similarity search use a filter-and-verify schema to
speed up query process. In the filtering phase, it computes the
candidate set $\id{Cand} = \{g:\xi(g, h) \leq~\tau \ and \ g \in G\}$,
where $\xi(g, h)$ is the lower bound on $\id{ged}(g,h)$. In the
verification phase, for each graph $g$ in~$\id{Cand}$, it needs
to compute $\id{ged}(g,h)$. Obviously, it is good for the size~$|\id{Cand}|$
as small as possible.

In this paper, we propose two filters, i.e., degree-based~$q$-$\id{gram}$
counting filter and degree-sequence filter using the degree structures
and label structures in a graph. Besides, we also use the following
two simple but effective global filters, i.e., number count filter~\cite{ZengTWFZ2009}
and label count filter~\cite{ZhaoXLW2012}. Number count filter is
derived based upon the differences of the number of vertices and
edges of comparing graphs and given by $\id{dist_N}(g, h) = ||\id{V_g}|-|\id{V_h}|| + ||\id{E_g}|-|\id{E_h}||$.
Label count filter improves the number count filter by taking labels
as well as structures into account and is given by $\id{dist_L}(g, h) = \idrm{max}\{|V_g|, |V_h|\}-|\Sigma_{\id{V_g}} \cap \Sigma_{\id{V_h}}|+
\idrm{max}\{|\id{E_g}|, |\id{E_h}|\}-|\Sigma_{\id{E_g}} \cap \Sigma_{\id{E_h}}|$.
By using all of them, we can obtain a candidate set as small as possible.

\section{Multiple Filters}
\label{sec:filters}

\subsection{Optimal Edit Path}

Given two graphs $g$ and $h$, and an optimal edit path~$P$ between
them, we group the operations on $P$ into five sets of edit operations:
vertex deletion group~$\id{P_{VD}}=\{p_i:p_i = (u \rightarrow~\epsilon) \in P\}$,
vertex insertion group $\id{P_{VI}}=\{p_i:p_i=(\epsilon \rightarrow~v) \in P\}$,
vertex substitution group $\id{P_{VS}} = \{p_i:p_i=(u \rightarrow v) \in P\}$,
edge deletion group $\id{P_{ED}}=\{p_i:p_i=(e(u, v) \rightarrow \epsilon) \ and \
((u \rightarrow \epsilon) \in \id{P_{VD}} \ or \ (v \rightarrow~\epsilon) \in \id{P_{VD}}) \}$
consists of the edge deletions performed on the deleted vertices,
and edge operation group $\id{P_O}$ consists of the edit operations
performed on edges except for those in $\id{P_{ED}}$.

For an optimal edit path $P$ between $g$ and $h$ , the insertion/deletion/substitution
edit operation on a vertex~$v$ or an edge $e$ must happen only
once, thus the edit operations in $\id{P_{VD}}$ are independent
of each other. Therefore, we can obtain an edit path by arbitrarily
arranging the edit operations in~$\id{P_{VD}}$. In the rest of this
paper, we use $\id{P_{VD}}$ to denote the edit operation set and the
edit operation sequence interchangeably when there is no ambiguity.
Similarly, $\id{P_{VI}}$, $\id{P_{VS}}$, $\id{P_{ED}}$ and $\id{P_O}$
could be also considered as the edit operation sets or paths. For
an optimal edit path $P$, we can always obtain an optimal edit
path $P'= \id{P_{ED}} \cdot \id{P_{VD}} \cdot \id{P_{VI}} \cdot \id{P_{VS}} \cdot \id{P_O}$
by arranging the edit operations in $P$. In the following section,
we consider $P = \id{P_{ED}} \cdot \id{P_{VD}} \cdot \id{P_{VI}} \cdot \id{P_{VS}} \cdot \id{P_O}$
as the default optimal edit path for two given graphs.

\begin{lem}
\label{lem:lem1}

Given graphs $g$ and $h$, and an optimal edit path $P$ that
transforms~$h$ to~$g$, then we have $|\id{P_{VD}}| = \idrm{max}\{|\id{V_h}|-|\id{V_g}|,0\}$
and $|\id{P_{VI}}| = \idrm{max}\{|\id{V_g}|-|\id{V_h}|,0\}$.

\end{lem}

\begin{proof}

Let $P = \id{P_{ED}} \cdot \id{P_{VD}} \cdot \id{P_{VI}} \cdot \id{P_{VS}} \cdot \id{P_O}$
be an edit optimal path that transforms $h$ to $g$. Then we
discuss the following three cases.

Case I.  When $|\id{V_h}| = |\id{V_g}|$. To transform $h$ to $g$,
the number of vertex deletions must be equal to that of vertex
insertions, i.e., $|\id{P_{VD}}| = |\id{P_{VI}}|$. We prove $|\id{P_{VD}}|
= |\id{P_{VI}}| = 0$ by contradiction. Assuming that $|\id{P_{VD}}|=|\id{P_{VI}}| = l \geq 1$,
thus there must exist at least one vertex insertion and one deletion.
Let $u$ be a deleted vertex and $v$ be a inserted vertex. We construct
another edit path $P'$ by $P$ as follows. First, we substitute
the label of~$u$ with $\mu(v)$, and then perform these edit operations
on $u$, which were performed on~$v$ before in~$\id{P_O}$. Finally,
we maintain the rest edit operations in $P$. In other words, we
replace $(u \rightarrow \epsilon)$ and $(\epsilon \rightarrow v)$
by $(u \rightarrow v)$. The length of~$P'$ is $|P'|=|P|-1<|P|$,
which contradicts the hypothesis that~$P$ is an optimal edit path.
Therefore, there exists no vertex deletions and insertions in~$P$,
i.e., $|P_{VD}|=|P_{VI}|=0$.

Case II. When $|\id{V_h}|<|\id{V_g}|$. There exists at least $|\id{V_g}|-|\id{V_h}|$
vertex insertions in~$P$. Let~$h_1$ be the graph obtained by
inserting $|\id{V_g}|-|\id{V_h}|$ vertices into~$h$. According
to the analysis in case I, no vertex deletions and insertions are
needed in an optimal edit path that transforms $h_1$ to~$g$. Thus,
$\idrm{only}$ $|\id{V_g}|-|\id{V_h}|$ vertex insertions are needed
in~$P$, i.e., $|\id{P_{VI}}|=|\id{V_g}|-|\id{V_h}|$ and $|\id{P_{VD}}|=0$.

Case III. When $|\id{V_h}|>|\id{V_g}|$. The proof is similar to the
proof of case II. We omit it here.
\end{proof}

\subsection{Q-gram Counting Filters}

\begin{defn}[Degree-based $q$-$\id{gram}$]
\label{def:degreeQgram}

Let $D_v =\\ (\mu(v), \id{adj}(v), d_v)$ be the degree structure
of vertex~$v$ in graph~$g$, where $\mu(v)$ is the label of~$v$,
$\id{adj}(v)$ is the multiset of labels for edges adjacent to~$v$
in~$g$, and $d_v$ is the degree of~$v$. The degree-based $q$-$\id{gram}$
set of graph~$g$ is defined as $D(g) = \{D_v: v \in V_g\}$.

\end{defn}

\begin{lem}
\label{lem:lem2}

Given two graphs $g$ and $h$, if $ged(g, h) \leq \tau$, then we
have $|D(g) \cap D(h)| \geq 2\idrm{max}\{|V_g|, |V_h|\}-|\Sigma_{V_g} \cap \Sigma_{V_h}| -2\tau$.
\end{lem}

\begin{proof}

First, we enumerate the effect of various edit operations on~$D(g)$:
(1) vertex insertion/deletion/substitution will affect one
degree-based $q$-$\id{gram}$. (2) edge insertion/deletion/\\substitution
will affect two degree-based $q$-$\id{grams}$. Then, without
loss of generality we assume that~$|V_h| \leq |V_g|$ and prove
Lemma~\ref{lem:lem2} as follows.

Let $P=\id{P_{ED}} \cdot \id{P_{VD}} \cdot \id{P_{VI}} \cdot \id{P_{VS}} \cdot \id{P_O}$
be an optimal edit path that transforms~$h$ to~$g$, such that:
$h \rightarrow h_1 \rightarrow g$, where~$h_1$
is obtained by performing $\id{P_{ED}} \cdot \id{P_{VD}} \cdot \id{P_{VI}} \cdot \id{P_{VS}}$
on~$h$, and $g$ is obtained by performing $\id{P_O}$ on $h_1$. By Lemma~\ref{lem:lem1},
we know that $|\id{P_{VD}}|=0$ and $|\id{P_{VI}}|=|V_g|-|V_h|$. Since
$\id{P_{ED}}$ consists of the edge deletions performed on the deleted
vertices, we have $|\id{P_{ED}}|=0$. To transform~$h$ to~$g$, $|\id{P_{VS}}|$ vertex
substitutions are needed, thus $|\id{P_{VS}}| \geq |V_g|-(|\Sigma_{V_g} \cap \Sigma_{V_h}|+ |\id{P_{VI}}|)=|V_h|-|\Sigma_{V_g} \cap \Sigma_{V_h}|$.
Since vertex insertion/substitution only affects one degree-based $q$-$\id{gram}$,
we have $|D(g) \cap D(h)| \geq |D(g) \cap D(h_1)| - (|\id{P_{VI}}| + |\id{P_{VS}}|)$.
Since~$\id{P_O}$ only consists of the edit operations performed
on edges and each of them affects two degree-based $q$-$\id{grams}$,
we have $|D(g) \cap D(h_1)| \geq |V_g|-2|\id{P_O}|$. Thus we have
$|D(g)\cap D(h)| \geq |V_g|-2|\id{P_O}|-(|\id{P_{VI}}| + |\id{P_{VS}}|)
\geq 2|V_g|-|\Sigma_{V_g} \cap \Sigma_{V_h}| -2\tau$.
\end{proof}

\begin{defn}[Label-based $q$-$\id{gram}$]
\label{def:labelQGram}

The label-based\\ $q$-$\id{gram}$ set of graph~$g$ is defined
as $L(g)=\Sigma_{V_g} \cup~\Sigma_{E_g}$, where~$\Sigma_{V_g}$
and $\Sigma_{E_g}$ are the label multisets of~$V_g$ and~$E_g$,
respectively.

\end{defn}

For the label-based $q$-$\id{gram}$, each edit operation $\idrm{affects}$
one $q$-$\id{gram}$, thus we can obtain the label-based $q$-$\id{gram}$
counting filter as follows. If $\id{ged}(g, h) \leq \tau$, then we have
$|L(g) \cap L(h)| \geq  \idrm{max}\{|V_g|, |V_h|\} + \idrm{max}\{|E_g|, |E_h|\} - \tau$.
It is a rewritten form of the label count filter~\cite{ZhaoXLW2012}.

Figure~\ref{fig:Fig03} shows the degree-based $q$-$\id{gram}$ and
label-based $q$-$\id{gram}$ sets of graphs shown in Figure~\ref{fig:Fig02}.
Note that the number on the left of each subgraph is the times of
the~$q$-$\id{gram}$ occurring in the graph and we omit the degree
value of each degree-based $q$-$\id{gram}$.

\begin{figure}[!htbp]
\centering
\begin{tikzpicture}[xshift=-1cm]
\tikzset{main node/.style={circle,draw,minimum size=0.3cm,inner sep=0pt},  }
\tiny{
   \node[main node, left] (1) at (0,0) [label=left:1] {A};
    \draw[line width=0.1mm, black]  (0,0) -- (0.3,0);
    \node[main node] (2) [below = 0.4cm  of 1,label=left:1] {C};
    \draw[line width=0.1mm, black]  (0,-0.7) -- (0.3,-0.7);
    \node[main node] (3) [below = 1.2cm  of 1,label=left:1] {A};
    \draw[line width=0.1mm, black]  (0,-1.5) -- (0.3,-1.4);
    \draw[line width=0.1mm, black]  (0,-1.5) -- (0.3,-1.6);

    \begin{scope}[xshift=1.1cm]
    \node[main node, left] (1) [label=left:3] {A};
    \draw[line width=0.1mm, black]  (0,0) -- (0.3,0);
    \node[main node] (2) [below = 1.2cm  of 1,label=left:1] {C};
    \draw[line width=0.1mm, black]  (0,-1.5) -- (0.3,-1.3);
    \draw[line width=0.1mm, black]  (0,-1.5) -- (0.3,-1.5);
    \draw[line width=0.1mm, black]  (0,-1.5) -- (0.3,-1.7);
    \end{scope}

    \begin{scope}[xshift=2.2cm]
    \node[main node, left] (1) [label=left:1] {C};
    \draw[line width=0.1mm, black]  (0,0) -- (0.3,0);
    \node[main node] (2) [below = 0.2cm  of 1,label=left:1] {B};
    \draw[line width=0.1mm, black]  (0,-0.5) -- (0.3,-0.4);
    \draw[line width=0.1mm, black]  (0,-0.5) -- (0.3,-0.6);
    \node[main node] (3) [below = 0.7cm  of 1,label=left:1] {C};
    \draw[line width=0.1mm, black]  (0,-1.0) -- (0.3,-0.9);
    \draw[line width=0.1mm, black]  (0,-1.0) -- (0.3,-1.1);
    \node[main node] (4) [below = 1.2cm  of 1,label=left:1] {A};
    \draw[line width=0.1mm, black]  (0,-1.5) -- (0.3,-1.3);
    \draw[line width=0.1mm, black]  (0,-1.5) -- (0.3,-1.5);
    \draw[line width=0.1mm, black]  (0,-1.5) -- (0.3,-1.7);
    \end{scope}

    \begin{scope}[xshift=3.3cm]
    \node[main node, left] (1) [label=left:2] {A};
    \draw[line width=0.1mm, black]  (0,0) -- (0.3,-0.1);
    \draw[line width=0.1mm, black]  (0,0) -- (0.3,0.1);
    \node[main node] (2) [below = 0.4cm  of 1,label=left:1] {B};
    \draw[line width=0.1mm, black]  (0,-0.7) -- (0.3,-0.6);
    \draw[line width=0.1mm, black]  (0,-0.7) -- (0.3,-0.8);
    \node[main node] (3) [below = 1.2cm  of 1,label=left:1] {C};
    \draw[line width=0.1mm, black]  (0,-1.5) -- (0.3,-1.4);
    \draw[line width=0.1mm, black]  (0,-1.5) -- (0.3,-1.6);
    \end{scope}

    \tkzText[below](0,-1.7){$g_1$}
    \tkzText[below](1.2,-1.7){$g_2$}
    \tkzText[below](2.3,-1.7){$g_3$}
    \tkzText[below](3.3,-1.7){$h$}

    \begin{scope}[xshift=5cm]
    \node[main node, left] (1) [label=left:2] {A};
    \node[main node] (2) [below = 0.4cm  of 1,label=left:1] {C};
    \draw [line width=0.1mm, black](-0.3,-1.5) -- (0.0,-1.5);
    \node[left] at (-0.3,-1.5) {2};
    \end{scope}

    \begin{scope}[xshift=5.7cm]
    \node[main node, left] (1) [label=left:3] {A};
    \node[main node] (2) [below = 0.4cm  of 1,label=left:1] {C};
    \draw [line width=0.1mm, black](-0.3,-1.5) -- (0.0,-1.5);
    \node[left] at (-0.3,-1.5) {3};
    \end{scope}

    \begin{scope}[xshift=6.4cm]
    \node[main node, left] (1) [label=left:1] {A};
    \node[main node] (2) [below = 0.2cm  of 1,label=left:1] {B};
     \node[main node] (3) [below = 0.7cm  of 1,label=left:2] {C};
    \draw [line width=0.1mm, black](-0.3,-1.5) -- (0.0,-1.5);
    \node[left] at (-0.3,-1.5) {4};
    \end{scope}

    \begin{scope}[xshift=7.1cm]
    \node[main node, left] (1) [label=left:2] {A};
    \node[main node] (2) [below = 0.2cm  of 1,label=left:1] {B};
    \node[main node] (3) [below = 0.7cm  of 1,label=left:1] {C};
    \draw [line width=0.1mm, black](-0.3,-1.5) -- (0.0,-1.5);
    \node[left] at (-0.3,-1.5) {4};
    \end{scope}
}

    \tkzText[below](4.9,-1.7){$g_1$}
    \tkzText[below](5.5,-1.7){$g_2$}
    \tkzText[below](6.2,-1.7){$g_3$}
    \tkzText[below](6.9,-1.7){$h$}
\end{tikzpicture}

\caption{\small{Degree-based $q$-$\id{gram}$ (left) and label-based $q$-$\id{gram}$ (right) sets.}}
\label{fig:Fig03}
\end{figure}
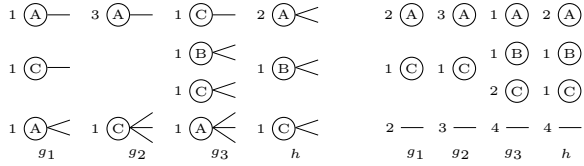

We use an example to illustrate the degree-based $q$-$\id{gram}$
and label-based $q$-$\id{gram}$ counting filters. For the graphs~$g_2$
and~$h$ shown in Figure~\ref{fig:Fig02}, if~$\tau$ = 2, by Lemma~\ref{lem:lem2}
we have $|D(g_2) \cap D(h)|= 0 < 2 \times \idrm{max}\{4, 4\} -|\{A, A, B, C\} \cap \{A, A, A, C\}|-2 \times 2 = 1$.
Thus, $g_2$ will be filtered out. However, for the graph $g_1$,
we have $|D(g_1) \cap D(h)|= 1 \geq  2 \times \idrm{max}\{4, 3\}-|\{A, A, B, C\} \cap \{A, A, C\}|-2 \times 2=1$
and hence $g_1$ will pass the filter. Similarly, only $g_1$ will
be filtered out by the label-based $q$-$\id{gram}$ counting filter
Therefore, we can filter~$g_1$ and~$g_2$ out using the degree-based
$q$-$\id{gram}$ and label-based~$q$-$\id{gram}$ counting filters.
However, for the graph~$g_3$ shown in Figure~\ref{fig:Fig02},
none of the above filters can filter it out. So, we propose another
filter, called degree-sequence filter, which utilizes the degrees
of vertices.

\subsection{Degree-Sequence Filter}
Let $\pi_g=[d_1, d_2,\ldots,d_{|V_g|}]$ be the degree vector
of graph~$g$, where~$d_i$ is the degree of vertex~$v_i$ in~$g$.
The degree sequence~$\sigma_g$ of~$g$ is a permutation of
$d_1, d_2,\ldots,d_{|V_g|}$ satisfying $\sigma_g[i] \geq \sigma_g[j]$
for~$i < j$. If~$g$ is isomorphic to~$h$, then we have~$\sigma_g=\sigma_h$.
Therefore, we can compute the lower bound on~$ged(g, h)$ using~$\sigma_g$
and $\sigma_h$.

\begin{defn}[Degree vector distance]
\label{def:degreeDistance} Given\\ two degree vectors~$\pi_g$
and $\pi_h$ such that $|\pi_g|$ = $|\pi_h|$. The distance between
them is defined as $\Delta(\pi_g, \pi_h) =\\
\lceil \sum_{\pi_h[i] \leq \pi_g[i]}(\pi_g[i]-\pi_h[i])/2 \rceil
+ \lceil \sum_{\pi_h[i] > \pi_g[i]}(\pi_h[i]-\pi_g[i])/2 \rceil$.

\end{defn}

\begin{lem}
\label{lem:lem3}

Let $x$ and $y$ be two degree vectors such that $|x|=|y|=n$ in
non-increasing order. For any bijection function $f:\{1, \ldots, n\} \rightarrow \{1, \ldots, n\}$,
we have $\Delta(x, z) \geq \Delta(x, y)$, where $z[i] = y[f(i)]$
for $1 \leq i \leq n$.

\end{lem}

\begin{proof}

For degree vectors~$x$ and~$y$, let $s_1(x,y)=\sum_{x[i] \leq y[i]}(y[i]-x[i])$
and $s_2(x,y)=\sum_{x[i] > y[i]} (x[i]-y[i])$. We have $\Delta(x,y)
=\lceil s_1(x,y)/2 \rceil+\lceil s_2(x,y)/2 \rceil$. We want to
prove $s_1(x, z) \geq s_1(x, y)$ and $s_2(x, z) \geq s_2(x, y)$,
where $z[i] = y[f(i)]$ for $1 \leq i \leq n$. We prove this claim
for~$x$ and~$z$ by induction on the vector length~$n$. And similar
claim holds for~$y$ and~$z$.

For the base case $n = 1$, it is trivial that $s_1(x, z) \geq s_1(x, y)$
and $s_2(x, z) \geq s_2(x, y)$. For the inductive step, we assume
that $s_1(x^k, z^k) \geq s_1(x^k, y^k)$ and $s_2(x^k, z^k) \geq s_2(x^k, y^k)$
for~$n \leq k$ where $x^k = [x[1], \ldots, x[k]]$.

We then prove the claim holds for $n = k + 1$. First, without
loss of generality, we assume that $f(k+1) = i\ (i < k + 1)$,
and $x[i] \geq y[i]$, thus we have $s_1(x^k, z^k) = s_1(x^{k-1}, z^{k-1})$
and $s_2(x^k, z^k) = s_2(x^{k-1}, z^{k-1}) + x[i] - y[i]$.
Then we consider the following three cases.

Case I. When $x[i] \geq x[k+1] \geq y[i] \geq y[k+1]$.

\vspace{-10pt}
\[
\begin{aligned}
s_1(x^{k+1}, z^{k+1}) & =   s_1(x^k, z^k)\\
                      &\geq s_1(x^k, y^k) = s_1(x^{k+1}, y^{k+1})\\
s_2(x^{k+1}, z^{k+1}) & =   s_2(x^{k-1}, z^{k-1}) + x[i]-y[i] + x[k+1]-y[k+1] \\
                      & = s_2(x^k, z^k) + x[k+1]-y[k+1] \\
                      &\geq s_2(x^k, y^k) + x[k+1]-y[k+1] \\
                      & = s_2(x^{k+1}, y^{k+1}).
\end{aligned}
\]

Case II. When $x[i] \geq y[i] \geq x[k+1] \geq y[k+1]$.

\vspace{-8pt}
\[
\begin{aligned}
s_1(x^{k+1}, z^{k+1}) &= s_1(x^{k-1}, z^{k-1}) + y[i] - x[k+1]\\
                      &= s_1(x^k, z^k) + y[i]-x[k+1]\\
                      &\geq s_1(x^k, y^k) = s_1(x^{k+1}, y^{k+1}).\\
s_2(x^{k+1}, z^{k+1}) &= s_2(x^{k-1}, z^{k-1}) + x[i]-y[k+1] \\
                      &= s_2(x^k, z^k) - (x[i] - y[i]) + x[i]-y[k+1] \\
                      &= s_2(x^k, z^k) + y[i] - y[k+1] \\
                      &\geq s_2(x^k, y^k) + x[k+1] - y[k+1] \\
                      &= s_2(x^{k+1}, y^{k+1}).
\end{aligned}
\]

Case III. When $x[i] \geq y[i] \geq y[k+1] \geq x[k+1]$.

\vspace{-8pt}
\[
\begin{aligned}
s_1(x^{k+1}, z^{k+1}) &= s_1(x^{k-1}, z^{k-1}) + y[i] - x[k+1]\\
                      &= s_1(x^k, z^k) + y[i] - x[k+1]\\
                      &\geq s_1(x^k, y^k) + y[k+1] - x[k+1] \\
                      &= s_1(x^{k+1}, y^{k+1}).\\
s_2(x^{k+1}, z^{k+1}) &= s_2(x^{k-1}, z^{k-1}) + x[i] - y[k+1]\\
                      &= s_2(x^k, z^k) - (x[i] - y[i]) + x[i] - y[k+1]\\
                      &= s_2(x^k, z^k) + y[i] - y[k+1]\\
                      &\geq s_2(x^k,y^k) = s_2(x^{k+1}, y^{k+1}).
\end{aligned}
\]

Finally, we note that
$s_1(x, z) \geq s_1(x, y) \ \idrm{and} \ s_2(x, z) \geq s_2(x, y)$,
and hence we have $\Delta(x, z) \geq \Delta(x, y)$.
\end{proof}

\begin{lem}
\label{lem:lem4}

Given two graphs~$g$ and~$h$ with~$|V_g|=|V_h|$, then we have $\id{ged}(g,h)
\geq \Delta(\sigma_g,\sigma_h)$.

\end{lem}

\begin{proof}

Let~$f$ be the bijection from the vertices in $h$ to that in~$g$
to ensure that the induced edit path is an optimal edit path.
Assuming that $\sigma_h[i]$ and $\sigma_g[f(i)]$ be the respective
degrees of a vertex~$v$ in $h$ and the corresponding vertex~$u$
in~$g$. If $\sigma_h[i] \leq \sigma_g[f(i)]$, we must insert at
least $(\sigma_g[f(i)]-~\sigma_h[i])$ edges on~$v$; otherwise,
we must delete at least $(\sigma_h[i]-\sigma_g[f(i)])$ edges.
Since one edge insertion/deletion affects degrees of two vertices,
we must insert at least $\lceil \sum_{\sigma_h[i] \leq \sigma_g[f(i)]} (\sigma_g[f(i)]-~\sigma_h[i])/2 \rceil$
edges. Similarly, we also need to delete at least $\lceil \sum_{\sigma_h[i]>\sigma_g[f(i)]}
(\sigma_h[i]-\sigma_g[f(i)])/2 \rceil$ edges. Thus, we have $\id{ged}(g, h)
\geq \Delta(\sigma_h, \pi_g^{'})$, where $\pi_g^{'}[i] = \sigma_g[f(i)]$
for $1\leq i \leq|V_h|$. By Lemma~\ref{lem:lem3}, we have $ged(g, h) \geq \Delta(\sigma_h, \pi_g^{'})
\geq \Delta(\sigma_g, \sigma_h)$.
\end{proof}

\begin{lem}[Degree-sequence filter]
\label{lem:lem5}

Given two \\ graphs $g$ and $h$, and an edit distance threshold~$\tau$,
if\\ $\id{ged}(g, h) \leq \tau$, then we have $\tau \geq \idrm{max}\{|\id{V_g}|,|\id{V_h}|\}-|\Sigma_{\id{V_g}} \cap \Sigma_{\id{V_h}}|+\lambda_e$,
where $$ \lambda_e = \left \{
\begin {array}{ll}
		\Delta(\sigma_g, \sigma_1) & \idrm{if} \; |\id{V_h}| \leq |\id{V_g}|;\\
        \idrm{min}_{h_1}\{|\id{E_h}|-\sum_j{\sigma_{h_1}[j]} + \Delta(\sigma_g, \sigma_{\id{h_1}})\} & \idrm{otherwise}. \\
\end{array} \right. $$
$\sigma_1=[\sigma_h[1],\ldots, \sigma_h[|\id{V_h}|], 0_1,\dots,0_{|\id{V_g}|-|\id{V_h}|}]$
and $h_1$ is a subgraph of~$h$ obtained by deleting $|\id{V_h}|-|\id{V_g}|$
vertices.

\end{lem}

\begin{proof}

Let $P=\id{P_{ED}} \cdot \id{P_{VD}} \cdot \id{P_{VI}} \cdot \id{P_{VS}} \cdot \id{P_O}$
be an optimal edit path that converts~$h$ to~$g$, satisfying
$h \rightarrow h_1 \rightarrow h_2 \rightarrow g$,
where~$h_1$ is obtained by performing $\id{P_{E_D}} \cdot \id{P_{V_D}}$
on~$h$, $h_2$ is obtained by performing $\id{P_{V_I}} \cdot \id{P_{V_S}}$
on~$h_1$ and $g$ is obtained by performing~$\id{P_O}$ on~$h_2$.
Then we discuss the following two cases.

Case I. When $|\id{V_h}| \leq |\id{V_g}|$. We have $|\id{P_{VD}}|=0$
and $|\id{P_{VI}}|=|\id{V_g}|-|\id{V_h}|$ and $|\id{P_{ED}}|= 0$
by Lemma~\ref{lem:lem1}. To transform~$h_1$ to~$h_2$, $|\id{P_{VS}}|$
vertex substitutions are needed in~$P$, thus we have
$|\id{P_{VS}}| \geq |\id{V_g}|-(|\id{P_{VI}}|+|\Sigma_{\id{V_g}}
\cap~\Sigma_{\id{V_h}}|) =|\id{V_h}|-|\Sigma_{\id{V_g}} \cap \Sigma_{\id{V_h}}|$.
Since $h_2$ is obtained by performing $\id{P_{ED}} \cdot \id{P_{VD}} \cdot \id{P_{VI}} \cdot \id{P_{VS}}$
on $h$, we have $\sigma_{h_2}=[\sigma_h[1],\ldots, \sigma_h[|\id{V_h}|], 0_1,\dots,0_{|\id{V_g}|-|\id{V_h}|}]$.
By Lemma~\ref{lem:lem4}, we have $|\id{P_O}|=\id{ged}(g, h_2) \geq \Delta(\sigma_g, \sigma_{h_2})$.
Therefore $\id{ged}(g,h)=|P|=|\id{P_{VI}}|+|\id{P_{VS}}|+|\id{P_O}| \geq |\id{V_g}|-|\Sigma_{\id{V_g}} \cap \Sigma_{\id{V_h}}|+ \Delta(\sigma_1, \sigma_g)$.

Case II. When $|\id{V_h}| > |\id{V_g}|$. We have~$|\id{P_{VI}}|=0$ and
$|\id{P_{VD}}|=|\id{V_h}|-|\id{V_g}|$ by Lemma~\ref{lem:lem1}. To
transform $h$ to $h_1$, the number of edge deletions in $\id{P_{ED}}$
is $|\id{P_{ED}}|=|\id{E_h}|-|E_{h_1}|=|\id{E_h}|-\sum_j\sigma_{h_1}[j]/2$.
Since only $|\id{P_{VS}}|$ vertex substitutions are needed to
transform~$h_1$ to~$h_2$, we have $|\id{P_{VS}}|
\geq |\id{V_h}|-(|\id{P_{VD}}|+|\Sigma_{\id{V_g}} \cap \Sigma_{\id{V_h}}|)
 =|\id{V_g}|-|\Sigma_{\id{V_g}} \cap \Sigma_{\id{V_h}}|$
and $\sigma_{h_1}=\sigma_{h_2}$. By Lemma~\ref{lem:lem4},
we also have $|\id{P_O}| \geq \Delta(\sigma_g, \sigma_{h_2})
=\Delta(\sigma_g, \sigma_{h_1})$.
Therefore $|P|=|\id{P_{ED}}|+|\id{P_{VD}}|+|\id{P_{VS}}|+|\id{P_O}|
\geq \idrm{min}_{h_1} \{|\id{V_h}|-|\Sigma_{\id{V_g}} \cap \Sigma_{\id{V_h}}|+ |\id{E_h}|-\sum_j\sigma_{h_1}[j]/2
+\Delta(\sigma_g, \sigma_{h_1})\}
=|\id{V_h}|-|\Sigma_{\id{V_g}} \cap \Sigma_{\id{V_h}}|+\idrm{min}_{h_1}\{|\id{E_h}|-
\sum_j\sigma_{h_1}[j]/2+\Delta(\sigma_g, \sigma_{h_1})\}$.
\end{proof}

We use an example to illustrate the degree-sequence filter. For
the graphs~$h$ and~$g_3$ shown in Figure~\ref{fig:Fig02}, we can
compute $\sigma_h=[2, 2, 2, 2]$ and $\sigma_{g_3} =[3, 2, 2, 1]$.
By Lemma~\ref{lem:lem5}, if $\tau = 2$, then we have $\idrm{max}\{4, 4\}-|\{A, A, B, C\} \cap \{A, B, C, C\}|+\Delta(\sigma_h,\sigma_{g_3})
=4-3+ \lceil(3-2)/2 \rceil + \lceil (2-1)/2 \rceil = 3 > 2$, then
we can filter~$g_3$ out.

\section{Reduced Query Region}
\label{sec:queryRegion}

Given a database $G$, we consider each graph~$g$ in~$G$ as a
point in the two-dimensional plane where the x-coordinate and
y-coordinate denote the number of vertices and edges in~$g$,
respectively. Thus the graph database~$G$ can be represented
as a set of points $S=\{(|\id{V_{g_j}}|, |\id{E_{g_j}}|):1 \leq j \leq~|G|\}$.
These points form a rectangle area $A = [x_{min}, x_{max}] \times [y_{min}, y_{max}]$,
where $x_{\id{min}} = \idrm{min}_{j}\{|V_{g_j}|\}, x_{\id{max}} = \idrm{max}_{j}\{|\id{V_{g_j}}|\},
\\ y_{\id{min}} = \idrm{min}_{j}\{|E_{g_j}|\}$ and $y_{\id{max}} = \idrm{max}_{j} \{|\id{E_{g_j}}|\}$
for $1\leq j \leq|G|$. By partitioning~$A$ into subregions,
we can perform a query at a reduced query region.

Given an initial division point $(x_0, y_0)$ and a length~$l$,
we partition~$A$ into disjoint subregions as follows. First,
we construct the initial square subregion $A_{0,0}$ formed by
the point set $\{(x, y):|x-x_0|+|y-y_0| \leq l\}$. Then, we extend
along the surrounding of $A_{0,0}$ to obtain subregions $A_{i,j}$
of the same size with $A_{0,0}$, where $i$ and $j$ denote the relative
offsets with respect to $A_{0,0}$ in lines $y=x$ and $y=-x$,
respectively. Finally, we repeat this process until all points
in~$A$ are exhausted. Then~$A$ is partitioned into some disjoint
subregions such that $A=\cup_{i, j}A_{i,j}$ and $A_{i, j} \cap A_{i',j'}
= \varnothing$ for all $i \neq i'$ and~$j \neq j'$.
Note that $i$ and $j$ can be negative.

\begin{defn}[Query rectangle and region]
\label{label:queryRectangle}

Given\\ a query graph $h$ and an edit distance threshold~$\tau$,
query rectangle~$A_h$ of $h$ is the rectangle formed by the point
set of $\{(x,y):|x-~|V_h||+|y-|E_h|| \leq \tau\}$. The query
region~$Q_h$ of~$h$ is the union of all subregions intersecting
with~$A_h$, i.e., $Q_h = \cup_{i, j}A_{i,j}$ such that~$A_{i,j}
\cap A_h \neq \varnothing $.

\end{defn}

For graphs $g$ and $h$, if $ged(g,h) \leq \tau$, then we have
$||V_g|-|V_h||+||E_g|-|E_h|| \leq \tau$. According to the definition
of~$A_h$, we know that $(|V_g|, |E_g|) \in A_h$. Since $Q_h= \cup_{i, j}A_{i,j}$ and $A_{i,j} \cap A_h \neq \varnothing $,
we have $A_h \subseteq Q_h$. Therefore we have $(|V_g|, |E_g|) \in Q_h$
and hence can reduce the query region from~$A$ to~$Q_h$. In the example
of Figure~\ref{fig:Fig04}, we have $Q_h=\{A_{0,0},A_{1,0},A_{0,-1},A_{1,-1}\}$
and then only need to perform the query at~$Q_h$.

\begin{figure}[htbp]
    \centering
    \includegraphics[width=3.0in]{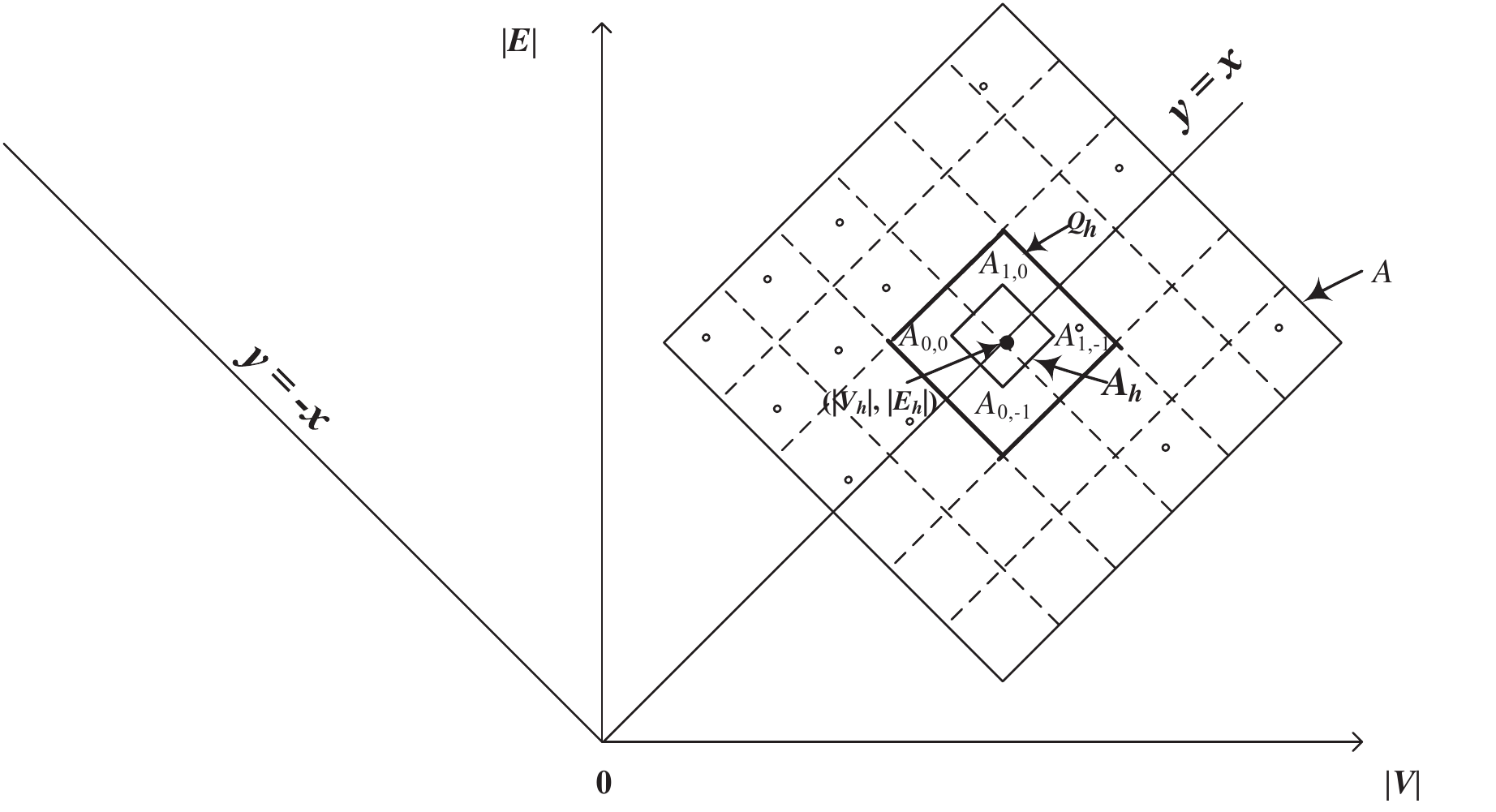}
    \caption{Illustration of $A_h$, $Q_h$ and $A$}
    \label{fig:Fig04}
\end{figure}

For a two-dimensional point $(x, y)$, its coordinates in lines
$y = x$ and $y = -x$ are $\frac{1}{\sqrt{2}}(x+y,y-x)$, thus its
relative offsets with respect to $(x_0, y_0)$ are $d_x=\frac{1}{\sqrt{2}}((x+y)-(x_0+y_0))$
and $d_y=\frac{1}{\sqrt{2}}((y-x)-(y_0-x_0))$ in $y=x$ and $y=-x$,
respectively. Since the side length of a subregion is $\frac{l}{\sqrt{2}}$,
the respective relative offsets with respect to $A_{0,0}$ are $\lfloor \frac{d_x}{l/\sqrt{2}} \rfloor$
and $\lfloor \frac{d_y}{l/\sqrt{2}} \rfloor$ in $y=x$ and $y=-x$.
Since the subregions in $Q_h$ are adjacent, we just need to find
the boundaries of subregions intersecting with $A_h$ using the
following formula.

\vspace{-10pt}
\begin{align}
\label{eq:eq1}
Q_h = \cup_{i, j}A_{i,j} \ \idrm{for \ all \ } i_1 \leq i \leq i_2 \ \idrm{and} \ j_1 \leq j \leq j_2.
\end{align}

\noindent where $i_1=\lfloor(|E_h| -\tau + |V_h|-(x_0 + y_0))/l\rfloor$
and $j_1 = \lfloor(|E_h| - \tau - |V_h|-(y_0-x_0))/l \rfloor$ are
the relative positions of the subregion in the lower left corner
of~$Q_h$ with respect to~$A_{0,0}$ in $y = x$ and $y = -x$,
respectively, $i_2 = \lfloor(|E_h| + \tau + |V_h|-(x_0 + y_0))/l \rfloor$
and $j_2 = \lfloor(|E_h| + \tau - |V_h|-(y_0-x_0))/l \rfloor$
are the respective relative positions of the subregion in the
top right corner of $Q_h$ with respect to $A_{0,0}$ in $y = x$
and $y = -x$.

\section{Succinct q-gram Tree Index}
\label{sec:Q-gramIndex}

Recall that we partitioned the region $A$ into some subregions
and then obtained a reduced query region $Q_h$. In order to
efficiently filter the graphs mapped into $Q_h$, we introduce
a space-efficient index structure via succinct representation of
the $q$-$\id{gram}$ tree as follows.

\subsection{Tree Structure}

Let $\mathcal{U}_D$ and $\mathcal{U}_L$ be the sets of all distinct
degree-based $q$-$\id{grams}$ and label-based $q$-$\id{grams}$
occurring in $G$, respectively, where $\id{\mathcal{U}_D(i)}$ and
$\id{\mathcal{U}_L(i)}$ are the~$i$th most frequently occurring
degree-based $q$-$\id{gram}$ and label-based $q$-$\id{gram}$
in~$G$, respectively. We use a four-tuple $\id{LD}=(\id{F_D}, \id{F_L}, n_v, n_e)$
to represent a graph~$g$, where~$n_v$ and~$n_e$ are the number
of vertices and edges in~$g$, respectively, $\id{F_D}$ and~$\id{F_L}$
are two arrays to store the degree-based $q$-$\id{gram}$ and label-based
$q$-$\id{gram}$ sets $D(g)$ and $L(g)$, respectively, where~$\id{F_D[i]}$
and~$\id{F_L[i]}$ are the respective number of occurrences of the
degree-based $q$-$\id{gram}$ $\id{\mathcal{U}_D(i)}$ in~$D(g)$ and
the label-based $q$-$\id{gram} \ \id{\mathcal{U}_L(i)}$ in~$L(g)$.

\begin{defn}
\label{def:four-union}
Given two four-tuples $\id{LD}$ and $\id{LD'}$, the union operator "$\sqcup$"
of~$\id{LD}$ and~$\id{LD'}$ is defined as: $\id{LD} \sqcup \id{LD'} = (\id{F_D} \oplus \id{F_D'}, \id{F_L} \oplus \id{F_L'}, \idrm{min}\{n_v, n_v'\},
\idrm{min}\{n_e, n_e'\})$, where
$$ (\id{F_D}\oplus \id{F_D'})[i] = \left \{
\begin {array}{ll}
	\idrm{max}\{\id{F_D[i]},\id{F_D'[i]}\}  & \idrm{if} \ i < \idrm{min}\{|\id{F_D}|, |\id{F_D'}|\}; \\
     \id{F_D[i]} 					 & \idrm{if} \ |\id{F_D'}| \leq i < |\id{F_D}|; \\
     \id{F_D'[i]} 					 & \idrm{if} \ |\id{F_D}| \leq i < |\id{F_D'}|.\\
\end{array} \right.\\ $$ and similar definition for $\id{F_L} \oplus \id{F_L'}$.
\end{defn}

Similarly, the union of multiple four-tuples can be defined
recursively.

\begin{defn}

A q-gram tree is a balanced tree such that each leaf node stores
the four-tuple~$\id{LD}$ of the data graph~$g$ and each internal
node is the union of its child nodes.

\end{defn}

Figure~\ref{fig:Fig05} gives an example of a $q$-$\id{gram}$
tree built on~$g_1$, $g_2$, and~$g_3$ shown in Figure~\ref{fig:Fig02}.

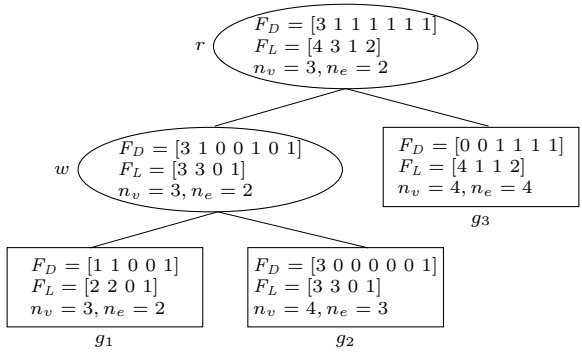
\begin{figure}[htbp]
\centering

\begin{tikzpicture}
\scriptsize{
\node[main node,shape=ellipse,minimum width=100pt,align=left] (1) at (0,0)[label=left:$r$]
     {$\id{F_D} = [3 \ 1 \ 1 \ 1 \ 1 \ 1 \ 1]$ \\ $\id{F_L} = [4 \ 3 \ 1 \ 2]$ \\  $n_v =3, n_e = 2$};
\node[main node, shape=ellipse,minimum width=100pt,align=left] (2) [xshift=-1.8cm, below = 0.5 of 1, label=left:$w$]
     {$\id{F_D} = [3 \ 1 \ 0 \ 0 \ 1 \ 0 \ 1]$ \\ $\id{F_L} = [3 \ 3 \ 0 \ 1]$ \\ $n_v =3, n_e = 2$};
\node[main node,shape=rectangle,align=left,minimum height=30pt,minimum width=74pt] (3) [xshift=1.8cm, below = 0.5 of 1, label=below:$g_3$]
     {$\id{F_D} = [0 \ 0 \ 1 \ 1 \ 1 \ 1]$ \\ $\id{F_L} = [4 \ 1 \ 1 \ 2]$ \\ $n_v =4, n_e = 4$};
\node[main node,shape=rectangle,align=left,minimum height=30pt,minimum width=74pt] (4) [xshift=-3.2cm, below = 2.1 of 1, label=below:$g_1$]
     {$\id{F_D} = [1 \ 1 \ 0 \ 0 \ 1]$ \\ $\id{F_L} = [2 \ 2 \ 0 \ 1]$ \\ $n_v =3, n_e = 2$};
\node[main node,shape=rectangle,align=left,minimum height=30pt,minimum width=74pt] (5) [xshift=0cm, below = 2.1 of 1, label=below:$g_2$]
     {$\id{F_D} = [3 \ 0 \ 0 \ 0 \ 0 \ 0 \ 1]$ \\ $\id{F_L} = [3 \ 3 \ 0 \ 1]$ \\ $n_v =4, n_e = 3$};

\draw [line width=0.1mm, black](0,-0.56) -- (-1.74,-1.06);
\draw [line width=0.1mm, black](0,-0.56) -- (1.9,-1.06);
\draw [line width=0.1mm, black](-1.7,-2.2) -- (-3.4,-2.68);
\draw [line width=0.1mm, black](-1.7,-2.2) -- (0.3,-2.68);
}
\end{tikzpicture}

\caption{\small Example of a $q$-$\id{gram}$ tree}
\label{fig:Fig05}
\end{figure}

\subsection{Succinct Representation}

The arrays $\id{F_D}$ and $\id{F_L}$ may contain lots of zeros,
thus a succinct representation of them is a space-efficient way
to store them. For a $q$-$\id{gram}$ tree, we obtain its succinct
representation by performing the following three steps. In the
following sections, we refer~$X$ to be~$D$ or~$L$.

(1) We use a bit vector $\id{I_X}$ and an array $\id{V_X}$ to
represent~$\id{F_X}$ as follows: if $\id{F_X[j]} = 0$ then we have
$\id{I_X[j]} = 0$; otherwise $\id{I_X[j]} = 1$. $\id{V_X[j]}$
represents the~$j$th nonzero entry in~$\id{F_X}$. For example,
the array~$\id{F_D}$ in the node $w$ shown in Figure~\ref{fig:Fig05}
is $\id{F_D}$ = [3 1 0 0 1 0 1], then we use ($\id{I_D}$, $\id{V_D}$)
= ([1 1 0 0 1 0 1], [3 1 1 1]) to represent $\id{F_D}$.

(2) We concatenate all bit vectors $\id{I_X}$ and arrays~$\id{V_X}$
for all nodes from the root node to leaves in a depth-first
traversal order to obtain a bit vector~$\id{B_X}$ and an array~$\id{\varPsi_X}$,
respectively. In addition, we also store the left and right
boundaries $\id{l_X}$ and $\id{r_X}$ of $\id{I_X}$ for each node,
respectively. For example, for the $q$-$\id{gram}$ tree shown in Figure~\ref{fig:Fig05},
we can obtain $\id{B_D}$ = [1 1 1 1 1 1 1 1 1 0 0 1 0 1 1 1 0 0 1 1 0 0 0 0 0 1 0 0 1 1 1 1]
and $\id{\varPsi_D}$ = [3 1 1 1 1 1 1 3 1 1 1 1 1 1 3 1 1 1 1 1].

(3) We divide $\id{\varPsi_X}$ into fixed-length blocks of size~$b$
and encode each block by choosing one from two different compression
methods so that the encoded bit vector~$\id{S_X}$ has the minimum
space. One compression method uses the fixed-length encoding
of $\lfloor \log b_{\id{max}} \rfloor +1$ bits to encode each entry
in a fixed-length encoding block, where $b_{\id{max}}$ is the
maximum value in this block. The other method uses Elias~$\gamma$
encoding to encode each entry in a $\gamma$-encoding block.
Logarithms in this paper are in base 2 unless otherwise stated.

To support random access to~$\id{\varPsi_X[j]}$, we also need to
store three auxiliary structures~$\id{SB_X}$, $\id{words_X}$,
and $\id{flag_X}$, where~$\id{SB_X}$ stores the starting position
of the encoding of each block in~$\id{S_X}$; the bit vector~$\id{flag_X}$
stores the encoding method used in each block such that $\id{flag_X}[k] = 1$
for the fixed-length encoding and $\id{flag_X}[k] = 0$ for
the Elias~$\gamma$ encoding for the $k$th block; $\id{words_X}$
stores the number of bits required for each entry in a fixed-length
encoding block. We also build rank dictionaries over the bit
vectors $\id{B_X}$ and~$\id{flag_X}$ to obtain $\id{rank}_1(\id{B_X}, j)$
and $\id{rank}_1(\id{flag_X}, j)$ in constant time~\cite{Jacobson1989},
where $\id{rank}_1(\id{B_X}, j)$ and $\id{rank}_1(\id{flag_X}, j)$
are the respective number of 1's up to $j$ in $\id{B_X}$ and $\id{flag_X}$.

Let $\id{B_D}$ and $\id{B_L}$ be the respective degree-based and
label-based $q$-$\id{grams}$ bit vectors, and $\id{\varPsi_D}$
and~$\id{\varPsi_L}$ be the respective degree-based and label-based
$q$-$\id{gram}$ frequency arrays. We use four structures $\id{S_D}$,
$\id{SB_D}$, $\id{flag_D}$, and~$\id{words_D}$ to represent~$\id{\varPsi_D}$.
Similarly, we use four structures $\id{S_L}$, $\id{SB_L}$, $\id{flag_L}$,
and $\id{words_L}$ to represent $\id{\varPsi_L}$. Figure~\ref{fig:Fig06}
shows the succinct representation of the $q$-$\id{gram}$ tree
shown in Figure~\ref{fig:Fig05}.

\begin{figure*}[htbp]
 \centering
	\begin{minipage}[b]{0.49\linewidth}
	\begin{tikzpicture}
        \scriptsize{
        \node[main node,shape=ellipse,minimum width=100pt,align=left] (1) at (0,0)[label=left:$r$]
             {$\id{l_D} =0, \id{r_D} = 6$  \\ $\id{l_L} =0,  \id{r_L} = 3$ \\ $n_v =3, n_e = 2$};
        \node[main node, shape=ellipse,minimum width=100pt, align=left] (2) [xshift=-1.8cm, below = 0.5 of 1, label=left:$w$]
             {$\id{l_D} =7, \id{r_D} = 13$ \\ $\id{l_L} =4,  \id{r_L} = 7$ \\ $n_v =3, n_e = 2$};
        \node[main node,shape=rectangle,align=left,minimum height=30pt,minimum width=74pt] (3) [xshift=1.8cm, below = 0.5 of 1, label=below:$g_3$]
             {$\id{l_D} =26, \id{r_D} = 31$ \\ $\id{l_L} =16, \id{r_L} = 19$ \\ $n_v =4, n_e = 4$};
        \node[main node,shape=rectangle,align=left,minimum height=30pt,minimum width=74pt] (4) [xshift=-3.4cm, below = 2.1 of 1, label=below:$g_1$]
             {$\id{l_D} =14, \id{r_D} = 18$ \\ $\id{l_L} =8,  \id{r_L} = 11$ \\ $n_v =3, n_e = 2$};
        \node[main node,shape=rectangle,align=left,minimum height=30pt,minimum width=74pt] (5) [xshift=-0.2cm, below = 2.1 of 1, label=below:$g_2$]
             {$\id{l_D} =19, \id{r_D} = 25$ \\ $\id{l_L} =12, \id{r_L} = 15$ \\ $n_v =4, n_e = 3$};

        \draw [line width=0.1mm, black](0,-0.56) -- (-1.74,-1.06);
        \draw [line width=0.1mm, black](0,-0.56) -- (1.9,-1.06);
        \draw [line width=0.1mm, black](-1.7,-2.2) -- (-3.2,-2.65);
        \draw [line width=0.1mm, black](-1.7,-2.2) -- (0.3,-2.65);
        }
   \end{tikzpicture}
   \centerline{(a)}
   \end{minipage}
   \begin{minipage}[b]{0.47\linewidth}
   \centering
    	\centerline{\includegraphics[width=1\textwidth]{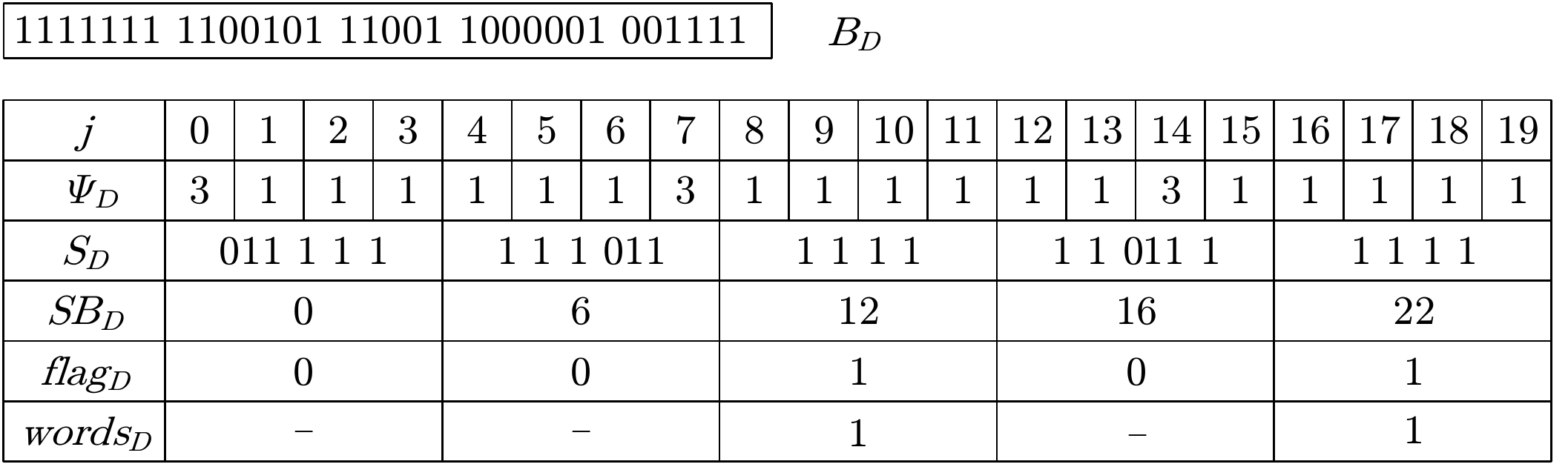}}
        \centerline{(b)}
        \vfill
        \centerline{\includegraphics[width=1\textwidth]{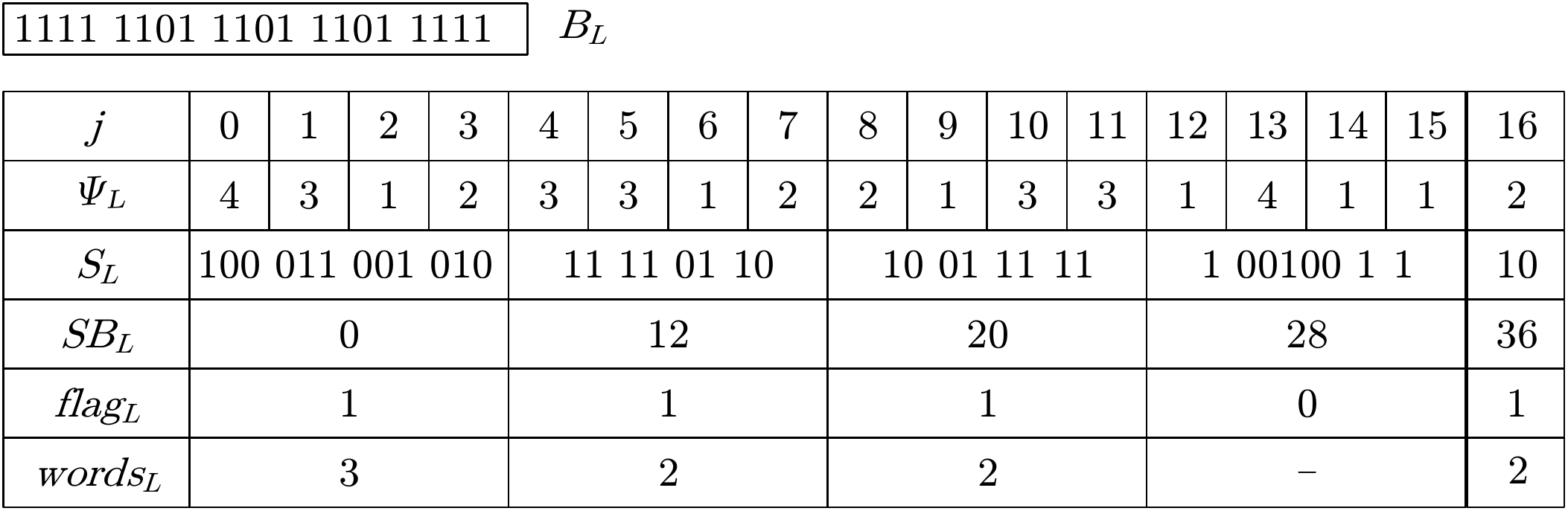}}
        \centerline{(c)}
   \end{minipage}
   \caption{Succinct representation of the $q$-$\id{gram}$ tree}
   \label{fig:Fig06}
\end{figure*}

\subsection{Access to $\varPsi_X$}

To access $\id{\varPsi_X[j]}$, we first query $\id{flag_X}$ and~$\id{SB_X}$
to determine the encoding method used and decoding position,
respectively, and then decode~$\id{S_X}$ from the decoding
position. The last decoded value is $\id{\varPsi_X[j]}$.

\vspace{-16pt}
\begin{multline}
\label{eq:eq2}
\id{\varPsi_X[j]}= decompress(\id{S_X}, \id{flag_X}[\lfloor j / b \rfloor],\ \id{SB_X}[\lfloor j / b \rfloor], {} \\
                (j \ \idrm{mod} \ b) + 1)
\end{multline}
\vspace{-14pt}

\noindent where $b$ is the block size. The operation $\id{decompress}$
performs a decoding on $\id{S_X}$. The encoding method and decoding
position are determined by the second parameter $\id{flag_X}[\lfloor j / b \rfloor]$
and third parameter $\id{SB_X}[\lfloor j / b \rfloor]$ of $\id{decompress}$,
respectively. ($j \mod b$) + 1 is the number of times needed to
be decoded.

For example, if we want to retrieve $\id{\varPsi_D}[14]$ (suppose that
the subscript starts from 0 and $b = 4$) shown in Figure~\ref{fig:Fig06},
we find that $\id{flag_D}[\lfloor 14/b \rfloor]=\id{flag_D}[3]=0$
and $\id{SB_D}[\lfloor 14/b \rfloor]=\id{SB_D}[3]=16$. Thus starting
from the~$16$th bit of $\id{S_D}$, we sequentially decode Elias~$\gamma$
encoding three times and the last decoded value is $\id{\varPsi_D}[14]=3$.

We use the following formula~(\ref{eq:eq3}) to compute the original
entry~$\id{F_X[i]}$ in the node~$w$.

\vspace{-10pt}
\begin{equation}
\id{F_X[i]} = \left \{
\begin {array}{ll}
	      0                              & \idrm{if} \id{B_X}[l_X + i] = 0;\\
    \id{\varPsi_X}[rank_1(B_X, l_X + i)] & \idrm{otherwise.} \\
\end{array} \right.
\label{eq:eq3}
\end{equation}

\noindent where $0 \leq i \leq \id{r_X}-\id{l_X}$, and $\id{l_X}$
and~$\id{r_X}$ are the left and right boundaries of $\id{I_X}$ for~$w$,
respectively. If $\id{B_X}[l_X + i] =~0$, then we have $\id{F_X}[i] = 0$
since $\id{I_X}[i] = \id{B_X}[l_X + i] = 0$; otherwise, we first
compute the position of $\id{F_X}[i]$ in $\id{\varPsi_X}$, i.e.,
$\id{rank}_1(\id{B_X}, \id{l_X} + i)$, and then use formula~(\ref{eq:eq2})
to compute $\id{\varPsi_X}[\id{rank}_1(\id{B_X}, l_X +~i)]$, where
$\id{rank}_1(\id{B_X}, \id{l_X} + i)$ is the number of 1's
up to~$\id{l_X} + i$ in~$\id{B_X}$, which can be computed in constant
time using a dictionary of $o(|\id{B_X}|)$ bits~\cite{Jacobson1989}.

For example, to retrieve $\id{F_D}[0]$ in the node~$g_2$ shown in
Figure~\ref{fig:Fig05}, we first obtain $\id{B_D}[l_D + 0] = \id{B_D}[19] = 1$,
and then compute its position in $\id{\varPsi_D}$ is $\id{rank}_1(\id{B_D}, 19) = 14$,
thus we have $\id{F_D}[0] = \id{\varPsi_D}[14] = 3$ by formula~(\ref{eq:eq2}).

As discussed above, the core operation in a succinct $q$-$\id{gram}$
tree is to calculate $\id{\varPsi_X}[j]$ by formula~(\ref{eq:eq2}),
i.e.,~the $\id{decompress}$ operation. In order to accelerate
the $\id{decompress}$ process, we use the look up table technique
proposed in~\cite{HuoCVY2014} to ensure that the $\id{decompress}$
operation takes a constant time.

\subsection{Space Analysis}
In this section we analyze the space occupied by the succinct
$q$-$\id{gram}$ tree $\id{T_{SQ}}$ built on $G$. $\id{T_{SQ}}$
consists of three parts: the respective index structures for
${\id{\varPsi_D}}$ and ${\id{\varPsi_L}}$, and left and right
boundaries, \#vertices and \#edges in each node of the tree.
The former contains encoded sequence $\id{S_X}$ and corresponding
auxiliary structures $\id{B_X}$, $\id{SB_X}$, $\id{flag_X}$ and
$\id{words_X}$; The latter consists of~$\id{l_X}$, $\id{r_X}$,
$n_v$ and $n_e$ stored in each node of $\id{T_{SQ}}$, where~$X$
denotes~$D$ or $L$. An illustration of these structures is shown
in Figure~\ref{fig:Fig06}.

Let $\id{v_m} = \idrm{max}_j\{\id{|V_{g_j}|}\}$, $\id{e_m} = \idrm{max}_j\{\id{|E_{g_j}|}\}$ for $1 \leq j \leq|G|$,
$\id{n_D} = |\id{B_D}|$, $\id{n_L} = |\id{B_L}|$, $b_D^{m}$ and
$b_L^{m}$ be the maximum value in $\id{\varPsi_D}$ and $\id{\varPsi_L}$,
respectively. For a degree-based $q$-$\id{gram}$, its maximum number
of occurrence in a graph $g$ is $|\id{V_g}|$, thus we have $b_D^{m} \leq \id{v_m}$.
Similarly, $b_L^{m} \leq \idrm{max}\{\id{v_m}, \id{e_m}\}$.
We first consider the space required by~$\id{l_X}$, $\id{r_X}$, $\id{n_v}$
and $\id{n_e}$.

For any node of $\id{T_{SQ}}$, we use $\lfloor \log \id{n_D} \rfloor +1$
bits to store $\id{l_D}$ and $\id{r_D}$, respectively, and $\lfloor \log n_L \rfloor +1$
bits to store $\id{l_L}$ and $\id{r_L}$, respectively, since $\id{l_D} \leq \id{r_D} \leq \id{n_D}$
and $\id{l_L} \leq \id{r_L} \leq \id{n_L}$. We also use respective
$\lfloor \log \id{v_m} \rfloor +1$ and $\lfloor \log \id{e_m} \rfloor +1$
bits to store $\id{n_v}$ and $\id{n_e}$, since $\id{n_v} \leq \id{v_m}$
and $\id{n_e} \leq \id{e_m}$. For an average fan-out of~$d$ for
each node in $\id{T_{SQ}}$ with $|G|$ leaf nodes, the total number
of nodes in~$\id{T_{SQ}}$ is bounded by $\sum_{h = 0}^{\log_{d}\id{|G|}}\frac{|G|}{d^h} \leq \frac{d|G|}{d-1}$.
Thus, we can use $\lfloor \log \frac{d|G|}{d-1} \rfloor +1$ bits
to store each child pointer of a node in~$\id{T_{SQ}}$. Thus, the
total number of bits required by $\id{l_D}$, $\id{r_D}$, $\id{l_L}$,
$\id{r_L}$, $\id{n_v}$, $\id{n_e}$ and pointers for all nodes in
$\id{T_{SQ}}$ is bounded by

\vspace{-10pt}
\[
\begin{aligned}
\frac{d|G|}{d-1}(2(\lfloor \log \id{n_D} \rfloor + 1) + 2(\lfloor \log \id{n_L} \rfloor + 1) + \lfloor \log \id{v_m} \rfloor + 1 +\\
              \lfloor \log \id{e_m} \rfloor + 1 + \lfloor \log \frac{d|G|}{d-1} \rfloor +1)\\
\leq \frac{d|G|}{d-1}(2\log (\id{n_D}\id{n_L}) + \log (\id{v_m}\id{e_m}) + \log \frac{d|G|}{d-1} + 7)&.
\end{aligned}
\]

We then consider the space required by $\id{S_X}$, $\id{B_X}$, $\id{SB_X}$,
$\id{flag_X}$ and $\id{words_X}$.

First we analyze the space needed by the encoded sequence $\id{S_X}$.
Let $N^g$ and $N^f$ be the respective collection of blocks with
$\gamma$ encoding and fixed-length encoding, and $|\gamma(b_i)|$ and
$|f(b_i)|$ be the respective number of bits needed to encode the
$i$th block $b_i$ using $\gamma$ encoding and fixed-length encoding.
By our hybrid encoding scheme, the number of bits required by~$\id{S_X}$
is bounded by

\vspace{-10pt}
\[
\begin{aligned}
   &\sum_{i=1}^{|\id{\varPsi_X}|/b} \idrm{min}\{|\gamma(b_i)|,|f(b_i)|\} \\
  =& \sum_{i\in \id{N^g}}|\gamma(b_i)| + \sum_{i\in\id{N^f}}|f(b_i)| \leq \sum_{i\in\id{N^g}}|f(b_i)| +\sum_{i\in\id{N^f}}|f(b_i)|\\
\leq& \sum_{i\in \id{N^g} \cup \id{N^f}}b(\lfloor \log b_{X}^{m} \rfloor +1) \leq \frac{|\id{\varPsi_X}|}{b}b(\lfloor \log b_{X}^{m} \rfloor +1) \\
\leq&\ |\id{\varPsi_X}|\log b_X^{m} + |\id{\varPsi_X}|.
\end{aligned}
\]
\noindent where the first inequality is due to the fact that
$|\gamma(b_i)|\leq |f(b_i)|$ when $i\in\id{N^g}$. The number
of bits required to encode block $b_i$ of $\id{\varPsi_X}$
using fixed-length encoding is bounded by $b(\lfloor \log b_{X}^{m} \rfloor +1)$.
The third inequality is due to the fact that $|N^g| + |N^f| = |\id{\varPsi_X}|/b$,
where $b$ is the block size.


Second, we analyze the space required by auxiliary structures
$\id{B_X}$, $\id{SB_X}$, $\id{flag_X}$ and $\id{words_X}$.

For bit vector $\id{B_X}$, the total number of bits required to
store it and its rank dictionary is $|\id{B_X}| + o(|\id{B_X}|)$
bits, where~$o(|\id{B_X}|)$ is the space in bits required by the
rank dictionary built on $\id{B_X}$ \cite{Jacobson1989}.

For $\id{SB_X}$, the space needed is $\frac{|\id{\varPsi_X}|}{b}(\log (|\id{\varPsi_X}|\log b_X^{m} + |\id{\varPsi_X}|) + 1)$
in bits in the worst case since each entry needs $\lfloor \log (|\id{\varPsi_X}|\log b_X^{m} + |\id{\varPsi_X}|) \rfloor + 1$
bits and there are $|\id{\varPsi_X}|/{b}$ blocks.

For $\id{flag_X}$, it is trivial that the total number of bits
required is $|\id{\varPsi_X}|/{b}+o(|\id{\varPsi_X}|/{b})$ bits,
since each block takes one bit and there are total $|\id{\varPsi_X}|/{b}$
blocks. The rank dictionary built on $\id{flag_X}$ needs
$o(|\id{\varPsi_X}|/{b})$ bits.

For $\id{words_X}$, the space used is bounded by
$\frac{|\id{\varPsi_X}|}{b} (\lfloor \log b_X^{m} \rfloor +~1)$,
since each entry requires $\lfloor \log b_X^{m} \rfloor +1$ bits
to store and there are~$|\id{\varPsi_X}|/{b}$ entries in the
worst case.

Putting all space needed for auxiliary structures $\id{B_X}$,
$\id{SB_X}$, $\id{flag_X}$ and $\id{words_X}$ together, we then
obtain

\vspace{-10pt}
\[
\begin{aligned}
|\id{B_X}| + o(|\id{B_X}|)+\frac{|\id{\varPsi_X}|}{b} \log(|\id{\varPsi_X}|\log b_X^{m} + |\id{\varPsi_X}|) +\\
 \frac{|\id{\varPsi_X}|}{b} \log b_X^{m} + 3 \frac{|\id{\varPsi_X}|}{b}+ o(\frac{|\id{\varPsi_X}|}{b})\\
= |\id{B_X}| + o(|\id{B_X}|)+o(|\id{\varPsi_X}|),\ \idrm{for}\ b = \log^{2}|\id{\varPsi_X}|.\qquad
\end{aligned}
\]

By adding $|\id{\varPsi_X}|\log b_X^{m} + |\id{\varPsi_X}|$ bits
required by~$\id{S_X}$ to the space required by auxiliary
structures, we obtain that the space is
$|\id{B_X}| + o(|\id{B_X}|)+ |\id{\varPsi_X}|\log b_X^{m} + |\id{\varPsi_X}| + o(|\id{\varPsi_X}|)$
bits.

By summing up all space for~$\id{T_{SQ}}$ and replacing~$X$ with~$D$ or~$L$,
we obtain that the succinct $q$-$\id{gram}$ tree $\id{T_{SQ}}$ takes
$\frac{d|G|}{d-1}(2\log(n_Dn_L) + \log(\id{v_m}\id{e_m}) + \log \frac{d|G|}{d-1} + 7) +
n_D + o(n_D)+ n_L + o(n_L)+|\id{\varPsi_D}|(\log\id{v_m} + 1) + |\id{\varPsi_L}|(\log\idrm{max}\{v_m, e_m\} + 1)
+ o(|\id{\varPsi_D}|) + o(|\id{\varPsi_L}|)$ bits of space.

\section{Query processing}
\label{sec:queryProcessing}

Our query process consists of two phrases. We first compute the
reduced query region~$\id{Q_h}$ by formula~(\ref{eq:eq1}), and then
perform the query on the succinct $q$-$\id{gram}$ trees built
on the graphs mapped into~$\id{Q_h}$.

\subsection{Query on Succinct q-gram Tree}
\label{sec:querySuccinctTree}

We introduce the query method on the succinct $q$-$\id{gram}$
tree~$T$ in this section.

\begin{lem}
\label{lem:lem6}

Let $\id{C_D}$ and $\id{C_L}$ be the respective number of common
degree-based and label-based $q$-$\id{grams}$ between any internal
node~$w$ of~$T$ and the query graph $h$, if $\id{C_D} < \idrm{max}\{n_v, |V_h|\}-2\tau$
or $\id{C_L} < \idrm{max}\{n_v, |V_h|\} +\idrm{max}\{n_e, |E_h|\}-~\tau$,
then we can safely prune all child nodes of $w$, where~$n_v$
and~$n_e$ are the number of vertices and edges in $w$, respectively.

\end{lem}

\begin{proof}

Let $\id{x.LD} = (\id{x.F_D}, \id{x.F_L}, x.n_v, x.n_e)$ denote
the four-tuple of~$x$, where $x$ is a node in $T$. For any internal
node~$w$ and a query graph $h$, we have $\id{C_D}=\sum_i \idrm{min}\{\id{w.F_D}[i], \id{h.F_D}[i]\}$.
According to Definition~\ref{def:four-union}, for a child node~$w_j$
of $w$ we have $\id{w_j.F_D}[i] \leq \id{w.F_D}[i]$. Therefore,
for a descendent leaf node (i.e., graph) $g$ of~$w$, we have
$|D(g) \cap D(h)| = \sum_i\idrm{min}\{\id{g.F_D}[i], \id{h.F_D}[i]\}
\leq \cdots \leq \sum_i\idrm{min}\{\id{w_j.F_D}[i], \id{h.F_D}[i]\} \leq \sum_i\idrm{min}\{\id{w.F_D}[i], \id{h.F_D}[i]\} =~\id{C_D}$.
Similarly, we also have $\id{w.n_v} \leq \id{w_j.n_v} \leq \cdots \leq \id{g.n_v}=|V_g|$
and $\id{w.n_e} \leq \id{w_j.n_e} \leq \cdots \leq \id{g.n_e}=|E_g|$.
If~$\id{C_D} < \idrm{max}\{\id{w.n_v}, |V_h|\}-2\tau$, then we have
$|D(g) \cap D(h)| \leq \id{C_D} < \idrm{max}\{\id{w.n_v}, |V_h|\}-2\tau
\leq \idrm{max}\{|V_g|, |V_h|\}-2\tau \leq 2\idrm{max}\{|V_g|, |V_h|\}
-|\Sigma_{V_g} \cap \Sigma_{V_h}|-2\tau$ and can safely prune
graph~$g$ by Lemma~\ref{lem:lem3}. Similarly, if $\id{C_L} < \idrm{max}\{\id{w.n_v}, |V_h|\} +
\idrm{max}\{\id{w.n_e}, |E_h|\}-\tau$, then we have $|L(g) \cap L(h)|
< \idrm{max}\{|V_g|, \\|V_h|\}  + \idrm{max}\{|E_g|, |E_h|\}-\tau$
and then can safely prune $g$ by the label-based $q$-$\id{gram}$
counting filter. So, we can safely prune all child nodes of~$w$.
\end{proof}

Algorithm~\ref{alg:searchQTree} gives the query algorithm on~$T$,
where~$r$ is the root node, $(\id{F_D}, \id{F_L}, n_v, n_e)$ is the
four-tuple of~$w$, $\id{l_D}$ and~$\id{r_D}$ are the left and right
boundaries of~$\id{I_D}$ for~$w$ in~$T$, respectively.

\begin{algorithm}
\caption{$\id{searchQTree}$($T, h, \tau$)}
\label{alg:searchQTree}
\DontPrintSemicolon
\KwIn{$ T, h, \tau $}
\KwOut{$\id{Cand}$}
$Cand \gets \; \varnothing $ \;
Compute the four-tuple $\id{LD'}=(\id{F_D'}, \id{F_L'}, n_v', n_e')$ and degree sequence $\sigma_h$ of $h$ \;
$searchTree$($r, \id{LD'}, \sigma_h, h, \tau$)\;
\Return $\id{Cand}$\;
\noindent \textbf{Procedure} $searchTree$($w, \id{LD'}, \sigma_h, h, \tau$)\;
$\id{C_L} \gets\; \sum_i \idrm{min}\{\id{F_L}[i], \id{F_L'}[i]\}$ \;
\If{$\id{C_L} \geq \idrm{max}\{n_v, |V_h|\} + \idrm{max}\{n_e, |E_h|\}-\tau$}{
     $\id{C_D} \gets\; \sum_i \idrm{min}\{\id{F_D}[i],\id{F_D'}[i]\}$\;
     \If{$\id{C_D} \geq \idrm{max}\{n_v, |V_h|\} - 2\tau$}{
         \If{$w$ is an internal node}{
           \For{each child $w_i$ of $w$}{
                \textbf{$searchTree$}($w_i, \id{LD'}, \sigma_h, h, \tau$)\;
           }
         }
         \If{$\id{C_D} \geq 2\idrm{max}\{|V_h|, |V_w|\}-|\Sigma_{V_w} \cap \Sigma_{V_h}|-2\tau$}{
             \For {$i \gets\; \id{l_D} \ \idrm{to} \ \id{r_D}$}{
                $\id{F_w}[i-l_D] \gets\; \id{F_D}[i-l_D]$\;
             }
             Obtain the degree sequence $\sigma_w$ of $w$\;
             Compute the lower bound $\xi$ \;
             \If {$\xi \leq \tau$}{
               $\id{Cand} \gets\; \id{Cand} \cup \{w\}$ \;
             }
        }
     }
}
\end{algorithm}

In Algorithm~\ref{alg:searchQTree}, we first compute the four-tuple
$LD'$ and degree sequence $\sigma_h$ of~$h$ in line~2, respectively,
and then perform the search processing $\id{searchTree}$ starting
from a node~$w$ initialized to~$r$, the root node of $T$ as follows.
First, we determine whether a node $w$ needs to be pruned based
upon Lemma~\ref{lem:lem6} in lines~6--12. In lines~6 and~8, we
compute the number of common label-based $q$-$\id{grams}$ $C_L$
and degree-based $q$-$\id{grams}$ $\id{C_D}$ between~$w$ and~$h$,
respectively. Note that, each entry $\id{F_D}[i]$ and $\id{F_L}[i]$
are compressed in $\id{\varPsi_D}$ and $\id{\varPsi_L}$, respectively,
thus we need to use formula (3) to compute them.
If $\id{C_L} < \idrm{max}\{n_v, |V_h|\} + \idrm{max}\{n_e, |E_h|\}-\tau$
or $\id{C_D} < \idrm{max}\{n_v, |V_h|\}-~2\tau$, we prune $w$; otherwise
each subtree of $w$ will be accessed in lines 11--12. Then, we
determine whether a node $w$ needs to be pruned based upon
Lemma~\ref{lem:lem2} in line~13, i.e., the degree-based $q$-$\id{gram}$
counting filter. If $\id{C_D} < 2\idrm{max}\{n_v, |V_h|\}-|\Sigma_{V_w} \cap \Sigma_{V_h}|-2\tau$,
we prune~$w$, where $|\Sigma_{V_w} \cap \Sigma_{V_h}|$ is the
number of common vertex labels between $w$ and $h$ obtained while
computing~$\id{C_L}$. Finally, we first obtain the degree sequence
$\sigma_w$ and then determine whether $w$ needs to be pruned
based upon Lemma~\ref{lem:lem5}, i.e., the degree-sequence filter.
In lines~14--15, we first compute the array~$F_w$ storing the
degree-based $q$-$\id{gram}$ set of $w$ and then obtain~$\sigma_w$
using $F_w$ and $\id{T_D}$ in line 16, where $\id{T_D}$ is a table
storing the mapping between a degree-based $q$-$\id{gram}$ and its
identifier. If $\xi > \tau$, then we prune $w$; otherwise, it passes
all filters to become a candidate.

\subsection{Query Algorithm}
\label{sec:whole}

Algorithm~\ref{alg:search} gives the whole query algorithm,
where~($x_0, y_0$) is the initial division point, $l$ is the
subregion length and $\id{T_{i, j}}$ is the succinct $q$-$\id{gram}$
tree built on the graphs mapped into the subregion~$\id{A_{i, j}}$.

\begin{algorithm}
\caption{$search$($h,\tau, x_0, y_0, l$)}
\label{alg:search}

\DontPrintSemicolon
\KwIn{$h,\tau, x_0, y_0, l$}
\KwOut{$\id{Cand}$}
$Cand \gets \; \varnothing $ \;
$Q_h \gets\; \cup_{i,j}\id{A_{i,j}}$ for all $i_1 \leq i \leq i_2$
and $j_1 \leq j \leq j_2$ \;
\ForEach{$\id{A_{i,j}} \subseteq Q_h $}{
        $\id{C_{i,j}} \gets\; \id{searchQTree}(\id{T_{i,j}}, h, \tau)$ \;
        $ \id{Cand} \gets \; \id{Cand} \cup \id{C_{i,j}} $ \;
}
\Return $\id{Cand}$
\end{algorithm}

In Algorithm~\ref{alg:search}, we first compute the query region~$Q_h$
by formula~(\ref{eq:eq1}) in line 2, where $i_1=\lfloor(|E_h| -\tau + |V_h|-(x_0 + y_0))/l\rfloor$,
$i_2 = \lfloor(|E_h| + \tau + |V_h|-(x_0 + y_0))/l \rfloor$,
$j_1 = \lfloor(|E_h| - \tau - |V_h|-(y_0-x_0))/l \rfloor$ and
$j_2 = \lfloor(|E_h| + \tau - |V_h|-(y_0-x_0))/l \rfloor$.
Then we only need to perform the query on the $q$-$\id{gram}$
trees~$\id{T_{i,j}}$ built on these subregions $\id{A_{i, j}}$ satisfying
$\id{A_{i, j}} \subseteq Q_h$ in lines~3--5. For each candidate graph~$g$
in~$\id{Cand}$, we can use the methods in~\cite{RiesenFB2007, RiesenEB2013, ZhaoXLW2012}
to compute the edit distance between~$g$ and~$h$ to seek for the
required graphs.

\section{Experimental results}
\label{sec:experiments}

In this section, we evaluate the performance of our proposed
method and compare it with C-Star~\cite{ZengTWFZ2009},
GSimJoin~\cite{ZhaoXLW2012} and Mixed~\cite{ZhengZLWZ2015}
on the real and synthetic datasets.

\subsection{Datasets and Settings}

We choose several real and synthetic datasets to test the performance
of the above approaches in our experiment, described as follows:

(1) AIDS\footnote{\url{http://dtp.nci.nih.gov/docs/aids/aids data.html}}.
It is a DTP AIDS antivirus screen compound dataset from the Development
and Therapeutics Program in NCI/NIH to discover compounds capable of
inhibiting the HIV virus. It contains 42687 chemical compounds. We
generate the labeled graphs from these chemical compounds and omit
Hydrogen atoms as did in~\cite{WangSHG2009}.

(2) PubChem\footnote{\url{http://pubchem.ncbi.nlm.nih.gov/}}. It
is a NIH funded project to record experimental data of chemical
interactions with biological systems. It contains more than 50
million chemical compounds until today. We randomly select 25
million chemical compounds to make up the large dataset PubChem-25M
used in this experiment.

(3) Synthetic. The synthetic datasets are generated by the synthetic
graph data generator GraphGen\footnote{\url{http://www.cse.ust.hk/graphgen/}}.
The synthetic generator can create a labeled and undirected graph
dataset. It allows us to specify various parameters, including the
dataset size, the average graph density $\rho = \frac{2|E|}{|V|(|V|-1)}$,
the number of edges in a graph, and the number of distinct vertex
and edge labels in the dataset, respectively. In order to evaluate
the performance of the above approaches on the density graphs, we
generate the dataset S100K.E30.D50.L5, which means that this dataset
contains 100000 graphs; the average density of each graph is 50\%;
the number of edges in each graph is 30; and the number of distinct
vertex and edge labels are 5 and 2, respectively.

For each dataset, we randomly select 50 graphs from it as its query
graphs. Table~\ref{tab:Statistics} summarizes some general characteristics
of the three datasets described above.

\begin{table}
\centering
\caption{Statistics of the three data sets}
\label{tab:Statistics}
\scriptsize
\begin{tabular}{|c|c|c|c|c|c|} \hline
DataSet          & $|G|$      & $|V|$ & $|E|$ & $|\Sigma_{V}|$ & $|\Sigma_{E}|$\\ \hline
AIDS             & 42687      & 25.6  & 27.5  & 62             & 3 \\ \hline
S100K.E30.D50.L5 & 100,000    & 11.02 & 30    & 5              & 2 \\ \hline
PubChem-25M      & 25,000,000 & 23.4  & 25.2  & 101            & 3 \\
\hline\end{tabular}
\end{table}

We have conducted all experiments on a HP Z800 PC with a 2.67 GHz CPU
and 24GB memory, running Ubuntu 12.04 operating system. We implemented
our algorithm in C++, with $-O3$ to compile and run. For GSimJoin, we
set $p = 4$ for the sparse graphs in datasets AIDS and PubChem-25M, and $p = 3$
for the density graphs in dataset S100K.E30.D50.L5, which are the recommended
values~\cite{ZhaoXLW2012}. In the following sections, we refer MSQ-Index
to our index structure and set the subregion length $l = 4$ and block size $b = 16$, respectively.

\subsection{Index Construction and Space Usage}

In this section, we introduce extensive experiments to evaluate
index construction performance of C-Star, GSimJoin, Mixed and MSQ-Index.

\subsubsection{Evaluating Our Index}

In order to evaluate the effectiveness of our hybrid encoding, we
compare it with fixed-length encoding, Elias~$\delta$ encoding,
Golomb encoding, and Elias~$\gamma$ encoding.

For each dataset described in Table~\ref{tab:Statistics}, we
show the number of bits on the average required by each entry
in $\id{\varPsi_D}$ and $\id{\varPsi_L}$ in Table~\ref{tab:AverageBits}
when applying fixed-length encoding ($f$), Elias~$\delta$ encoding
($\delta$), Golomb encoding ($g$), Elias~$\gamma$ encoding ($\gamma$)
and hybrid encoding ($h$) to $\id{\varPsi_D}$ and $\id{\varPsi_L}$,
where S100K and Pub-25M stand for S100K.E30.D50.L5 and PubChem-25M,
respectively. Each entry in  $\id{\varPsi_D}$ and $\id{\varPsi_L}$
uses about between~4 and 6 bits on the tested data, which is much
smaller than that used to represent an entry in the previous
state-of-the-art indexing methods compared in this paper.

\begin{table}[ht]
\centering
\caption{Average number of bits}
\label{tab:AverageBits}
\scriptsize
\begin{tabular}{|p{0.06\textwidth}<{\centering}|p{0.02\textwidth}<{\centering}|p{0.02\textwidth}<{\centering}|p{0.02\textwidth}<{\centering}
                |p{0.02\textwidth}<{\centering}|p{0.02\textwidth}<{\centering}|p{0.02\textwidth}<{\centering}|p{0.02\textwidth}<{\centering}
                |p{0.02\textwidth}<{\centering}|p{0.02\textwidth}<{\centering}|p{0.02\textwidth}<{\centering}|p{0.02\textwidth}<{\centering}
                |p{0.02\textwidth}<{\centering}|p{0.02\textwidth}<{\centering}|}

  \hline
  \multirow {2}{*}{$Datasets$} & \multicolumn{5}{c|}{$\varPsi_D$}   &  \multicolumn{5}{c|}{$\varPsi_L$}  \\ \cline{2-11}
                               & $f$  & $g$  & $\delta$ & $\gamma$  & $h$           & $f$   & $g$  & $\delta$ & $\gamma$  & $h$  \\  \hline
  AIDS                         & 4.51 & 4.25 & 3.68     & 3.57      & \textbf{3.51} & 5.97  & 6.23 & 6.17     & 6.13      & \textbf{5.93} \\   \hline
  S100K                        & 3.33 & 4.04 & 3.31     & 3.18      & \textbf{3.08} & 4.01  & 4.75 & 4.5      & 4.27     & \textbf{3.88} \\  \hline
  Pub-25M                      & 4.57 & 4.19 & 3.62     & 3.61      & \textbf{3.36} & 5.73  & 6.07 & 5.85     & 5.79      & \textbf{5.31} \\  \hline
\end{tabular}
\end{table}

Among all the encoding methods shown in Table~\ref{tab:AverageBits},
hybrid encoding gives the minimum space. Compared with fixed-length
encoding, the average number of bits required for hybrid encoding
decreases by about 10\%.

\begin{table}[htbp]
\centering
\caption{Storage space (MByte)}
\label{tab:indexSize}
\scriptsize
\begin{tabular}{|p{0.06\textwidth}<{\centering}|p{0.04\textwidth}<{\centering}|p{0.046\textwidth}<{\centering}|p{0.046\textwidth}<{\centering}
                |p{0.03\textwidth}<{\centering}|p{0.04\textwidth}<{\centering}|p{0.04\textwidth}<{\centering}|}
  \hline
  \multirow {2}{*}{$Datasets$} & \multicolumn{3}{c|}{$T_Q$}& \multicolumn{3}{c|}{$T_{SQ}$}         \\ \cline{2-7}
                               & $S_a$   & $S_b$   & $S_c$   & $S_a'$   & $S_b'$  & $S_c'$    \\ \hline
  AIDS                         & 0.29    & 6.09    & 2.11    & 0.51        & 0.39       & 0.27       \\ \hline
  S100K                        & 0.53    & 5.45    & 7.78    & 1.13        & 0.29       & 0.44       \\ \hline
  Pub-25M                      & 343.58  & 3669.91 & 2014.12 & 598.1       & 303.83     & 237.21      \\ \hline
\end{tabular}
\end{table}

In Table~\ref{tab:indexSize}, we report the storage space of
the $q$-$\id{gram}$ tree $\id{T_Q}$ and its succinct representation
$\id{T_{SQ}}$ built on the above three datasets. For
$\id{T_Q}$, we decompose its storage space into three parts
$S_a$, $S_b$ and $S_c$, where $S_a$ is the storage space of
$n_v, n_e$ and pointers of all nodes, and $S_b$ and $S_c$ are
the storage space of $\id{F_D}$ and $\id{F_L}$ of all nodes,
respectively. Correspondingly, $S_a'$ is the storage space of
$\id{n_v}$, $\id{n_e}$, $\id{l_D}$, $\id{r_D}$, $\id{l_L}$,
$\id{r_L}$ and pointers of all nodes, shown in Figure~\ref{fig:Fig06}(a).
$S_b'$ is total storage space of $\id{B_D}$, $\id{S_D}$, $\id{SB_D}$, $\id{words_D}$
and $\id{flag_D}$, shown in Figure~\ref{fig:Fig06}(b), and
$S_c'$ is total storage space of $\id{B_L}$, $\id{S_L}$,
$\id{SB_L}$, $\id{words_L}$ and $\id{flag_L}$, shown in
Figure~\ref{fig:Fig06}(c).

From Table~\ref{tab:indexSize}, we know that $S_b$ and $S_c$
take up most amount of storage space of $\id{T_Q}$, thus a
succinct representation of $\id{F_D}$ and $\id{F_L}$ of all
nodes is an efficient way to reduce the storage space of $\id{T_Q}$.
Compared with $S_b$ and $S_c$, both $S_b'$ and $S_c'$ can
be reduced by more than 90\%. This is because that (1) only
nonzero entries are needed to encode in the succinct representation;
(2) our hybrid encoding will greatly reduce the number of bits
required for each nonzero entry. Compared with the storage space
of $\id{T_Q}$ (the sum of $S_a$, $S_b$ and $S_c$), the storage
space of $\id{T_{SQ}}$ (the sum of $S_a'$, $S_b'$ and
$S_c'$) can be reduced by more than 80\%. Thus, the succinct
representation of $q$-$\id{gram}$ tree can greatly reduce the
storage space.

\begin{figure*}[ht]
\centering
    \begin{minipage}{0.28\linewidth}
		\centerline{\includegraphics[width=1\textwidth]{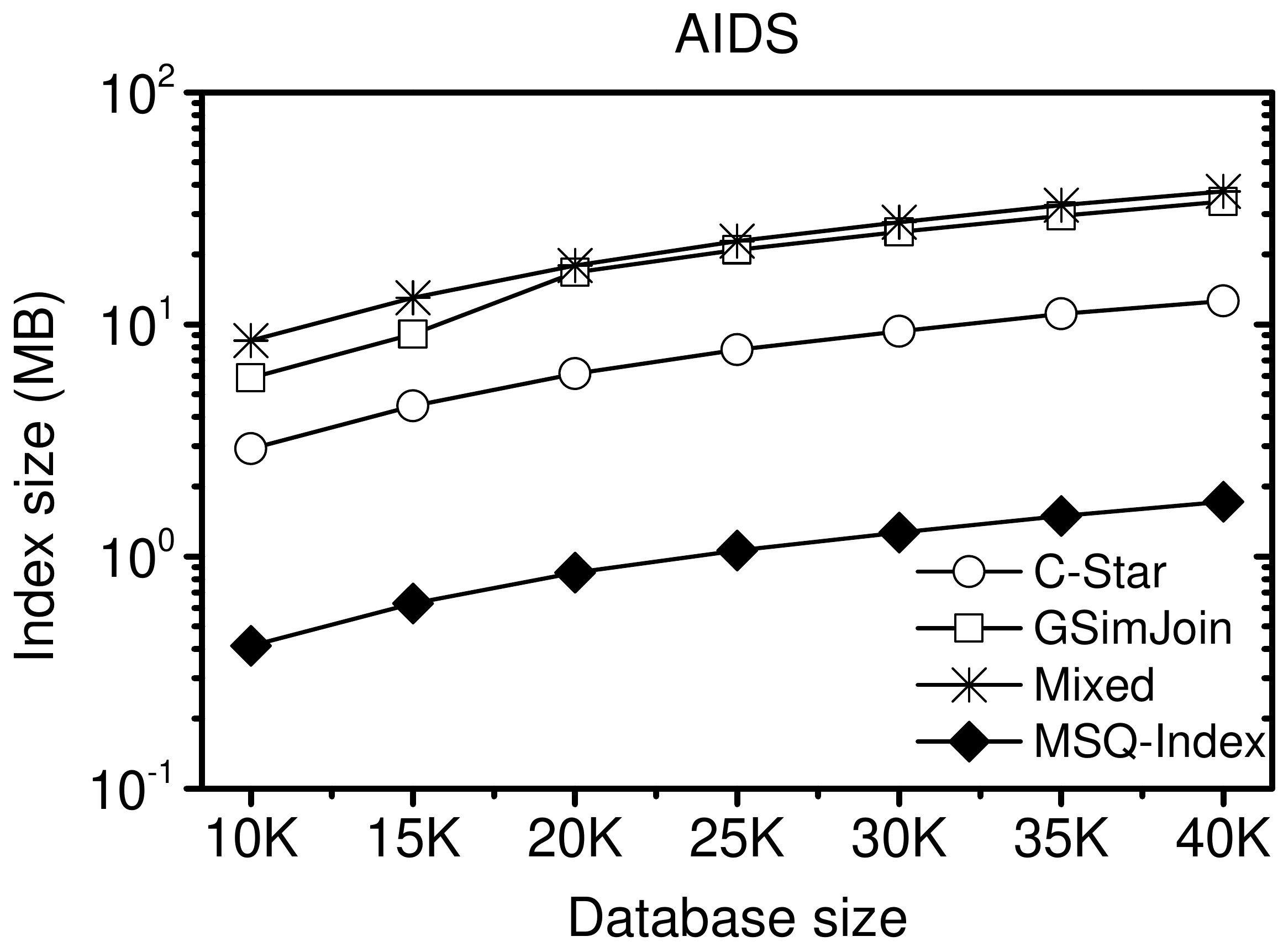}}
	\end{minipage}
	\qquad
	\begin{minipage}{0.28\linewidth}
		\centerline{\includegraphics[width=1\textwidth]{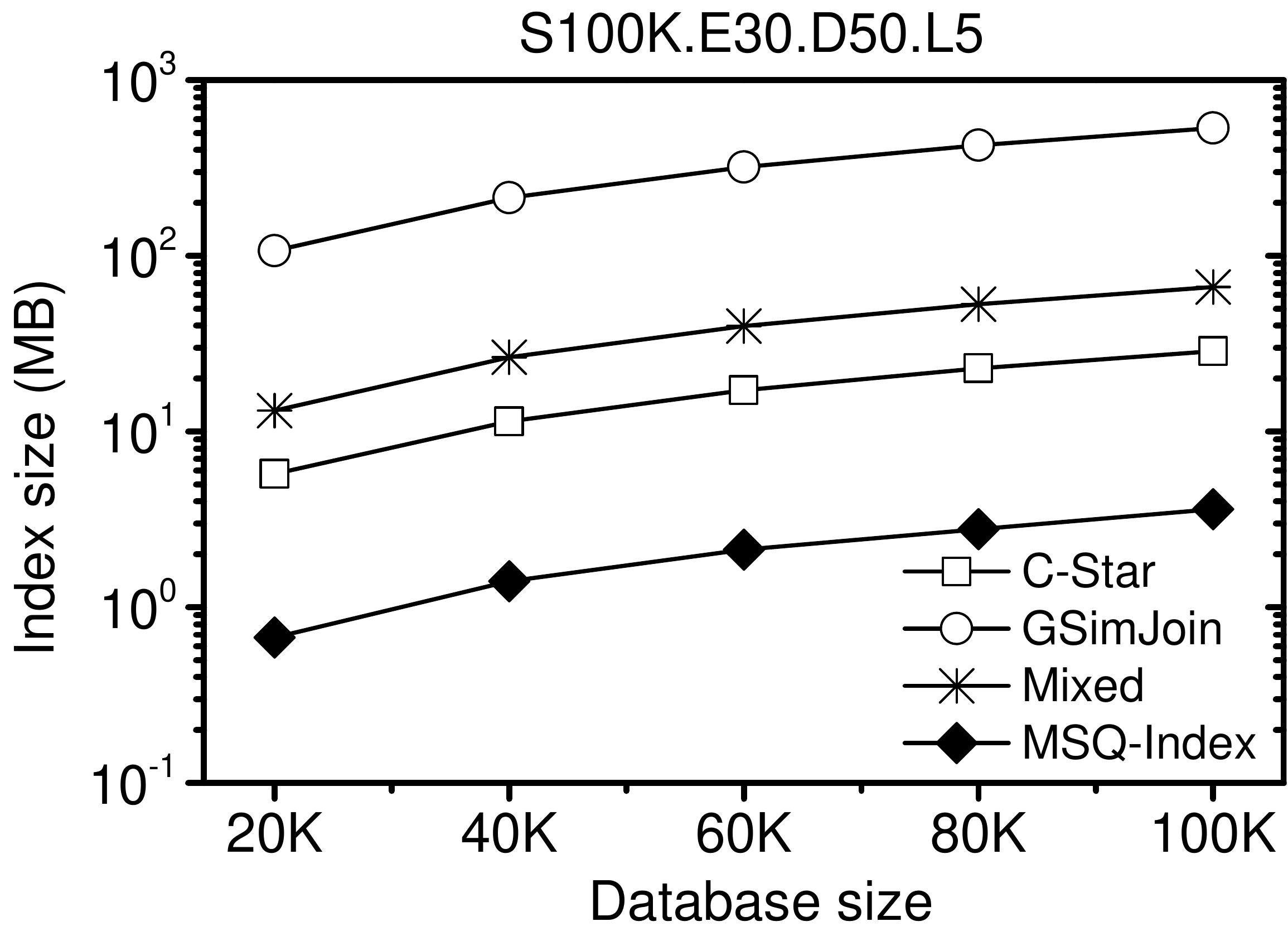}}
	\end{minipage}
	\qquad
	\begin{minipage}{0.28\linewidth}
		\centerline{\includegraphics[width=1\textwidth]{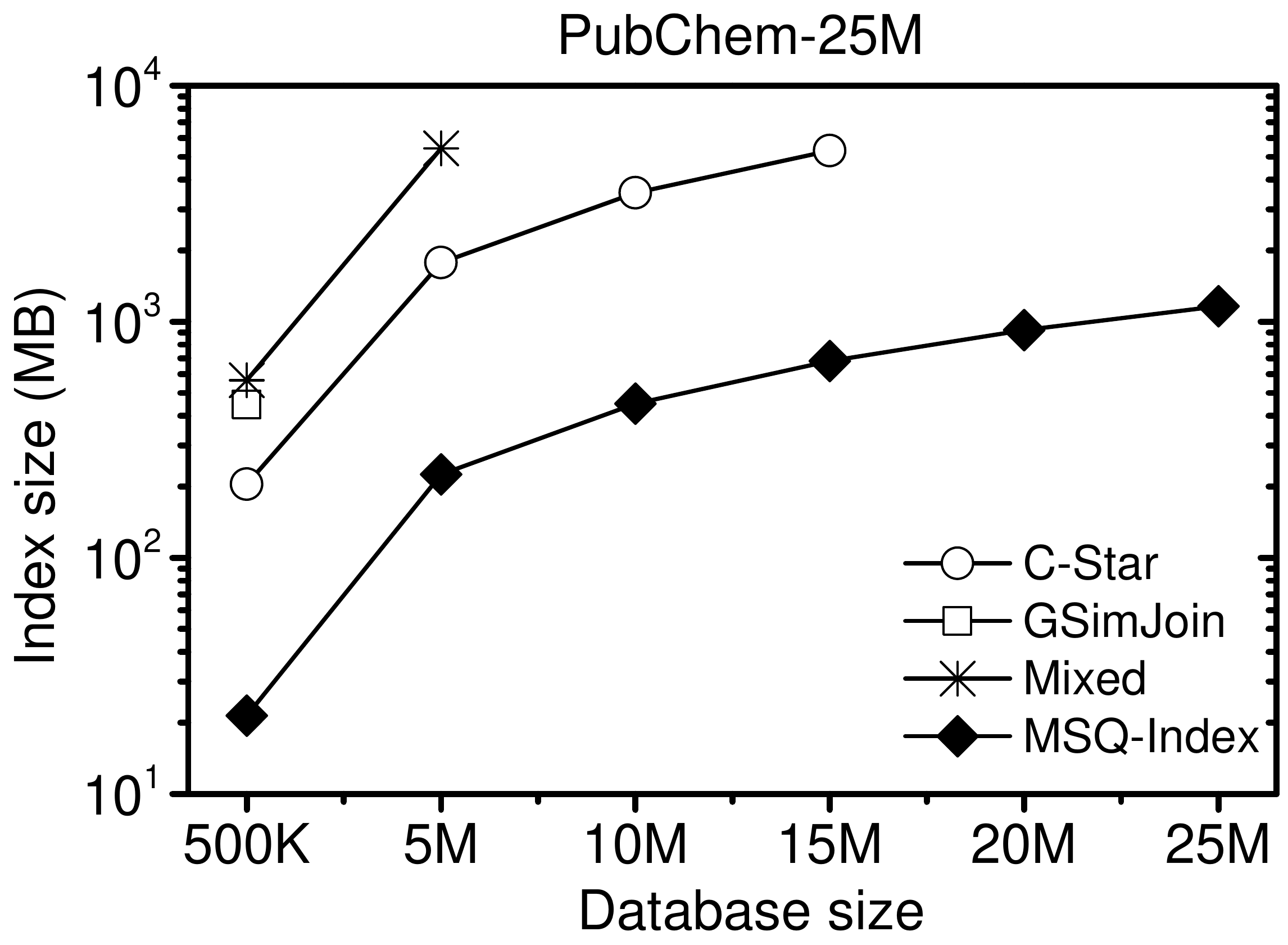}}
    \end{minipage}
    \vfill
    \begin{minipage}{0.28\linewidth}
		\centerline{\includegraphics[width=1\textwidth]{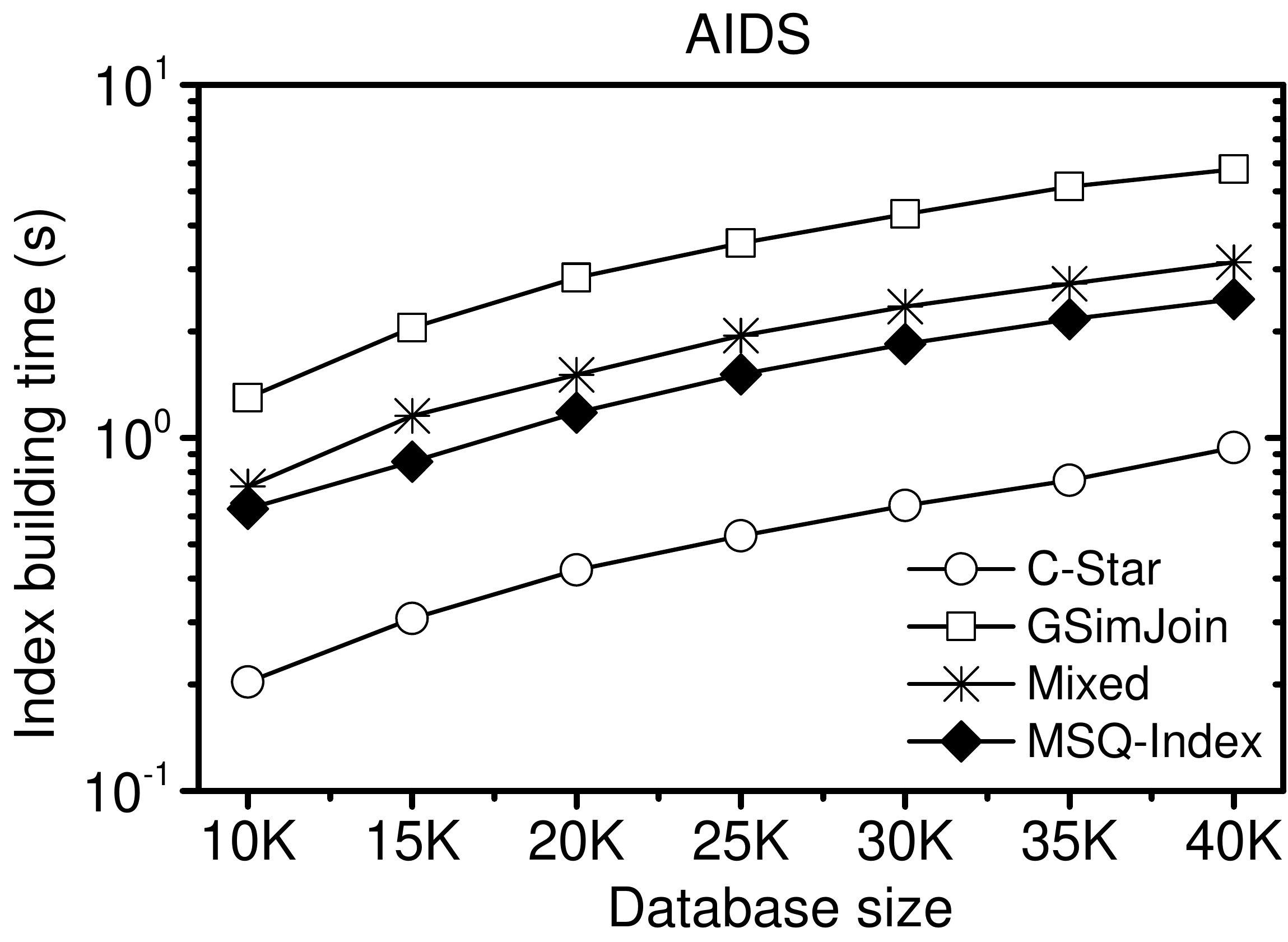}}
	\end{minipage}
	\qquad
	\begin{minipage}{0.28\linewidth}
		\centerline{\includegraphics[width=1\textwidth]{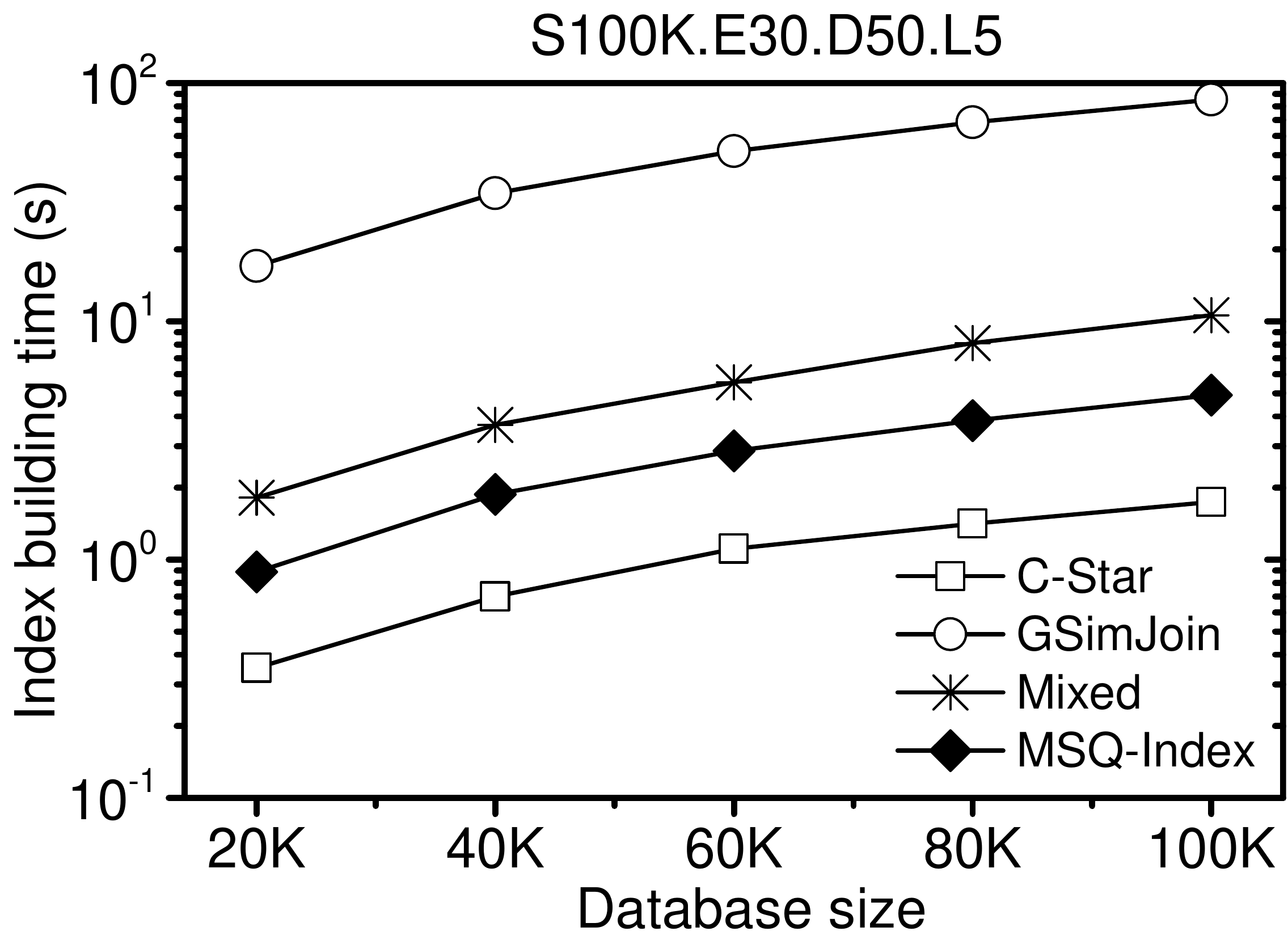}}
	\end{minipage}
	\qquad
	\begin{minipage}{0.28\linewidth}
		\centerline{\includegraphics[width=1\textwidth]{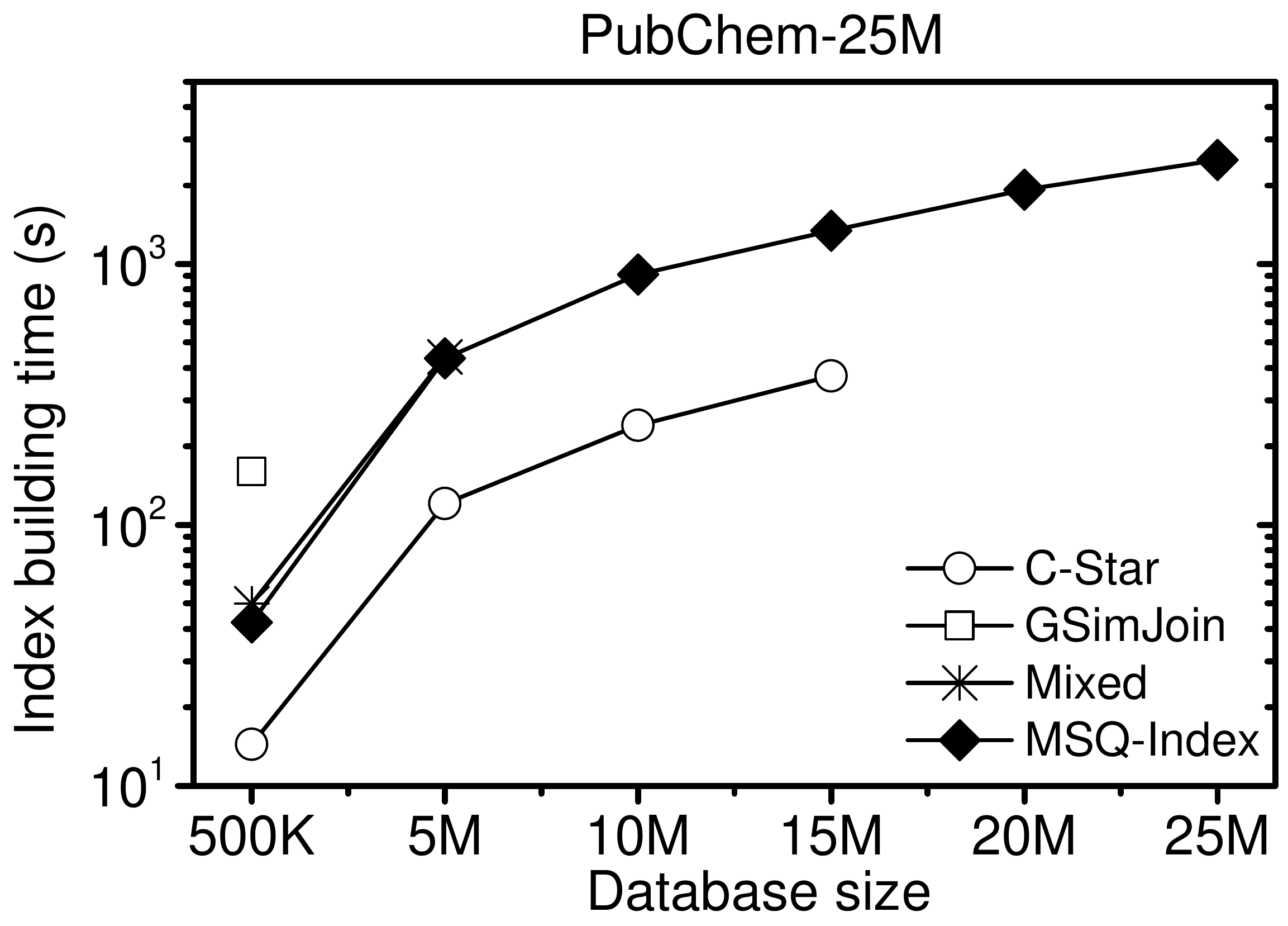}}
	\end{minipage}
\caption{Index size and building time on the three datasets.}
\label{fig:Fig08}
\end{figure*}

\begin{figure*}[htbp]
\centering
    \begin{minipage}{0.28\linewidth}
		\centerline{\includegraphics[width=1\textwidth]{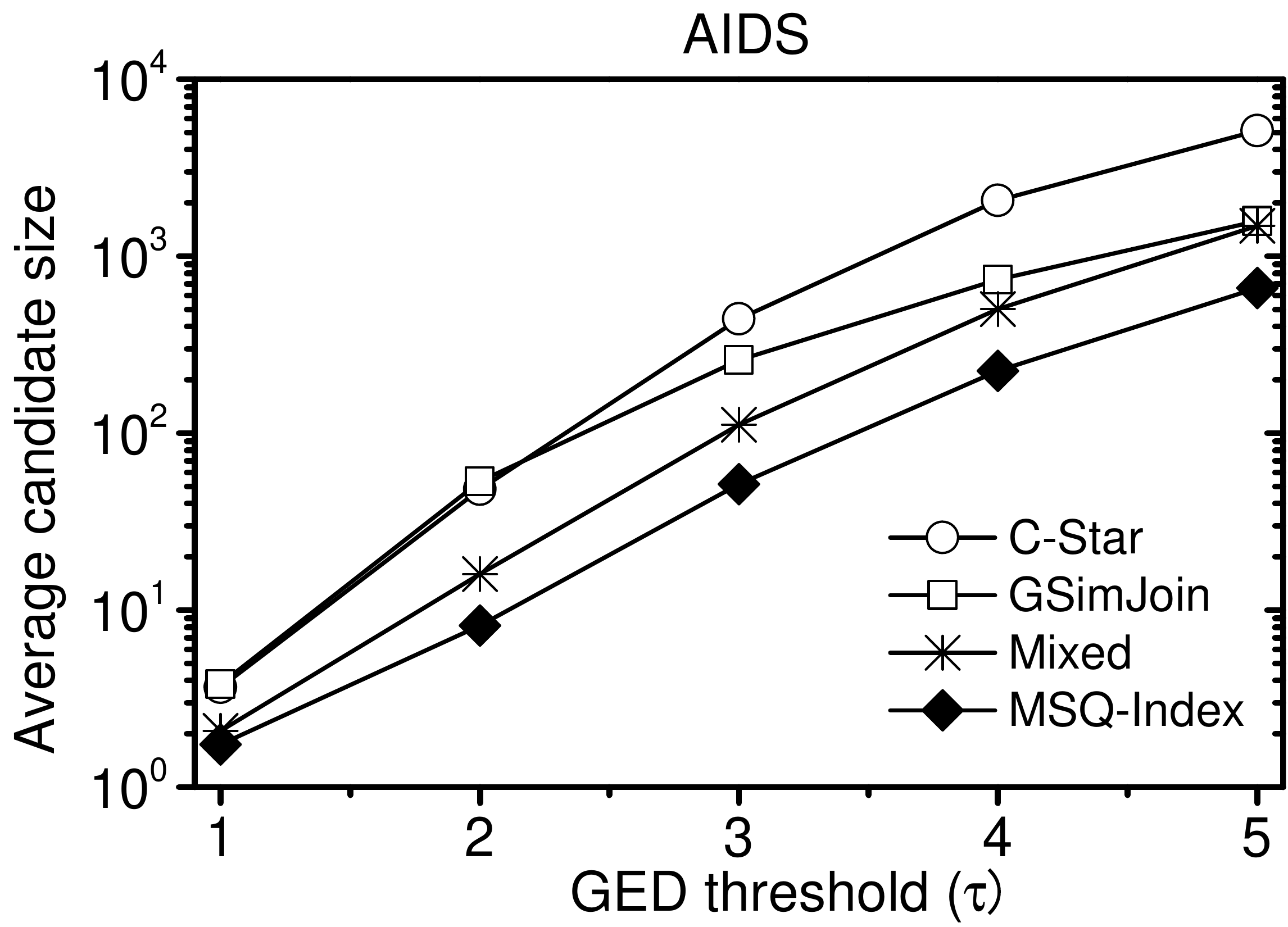}}
	\end{minipage}
	\qquad
	\begin{minipage}{0.28\linewidth}
		\centerline{\includegraphics[width=1\textwidth]{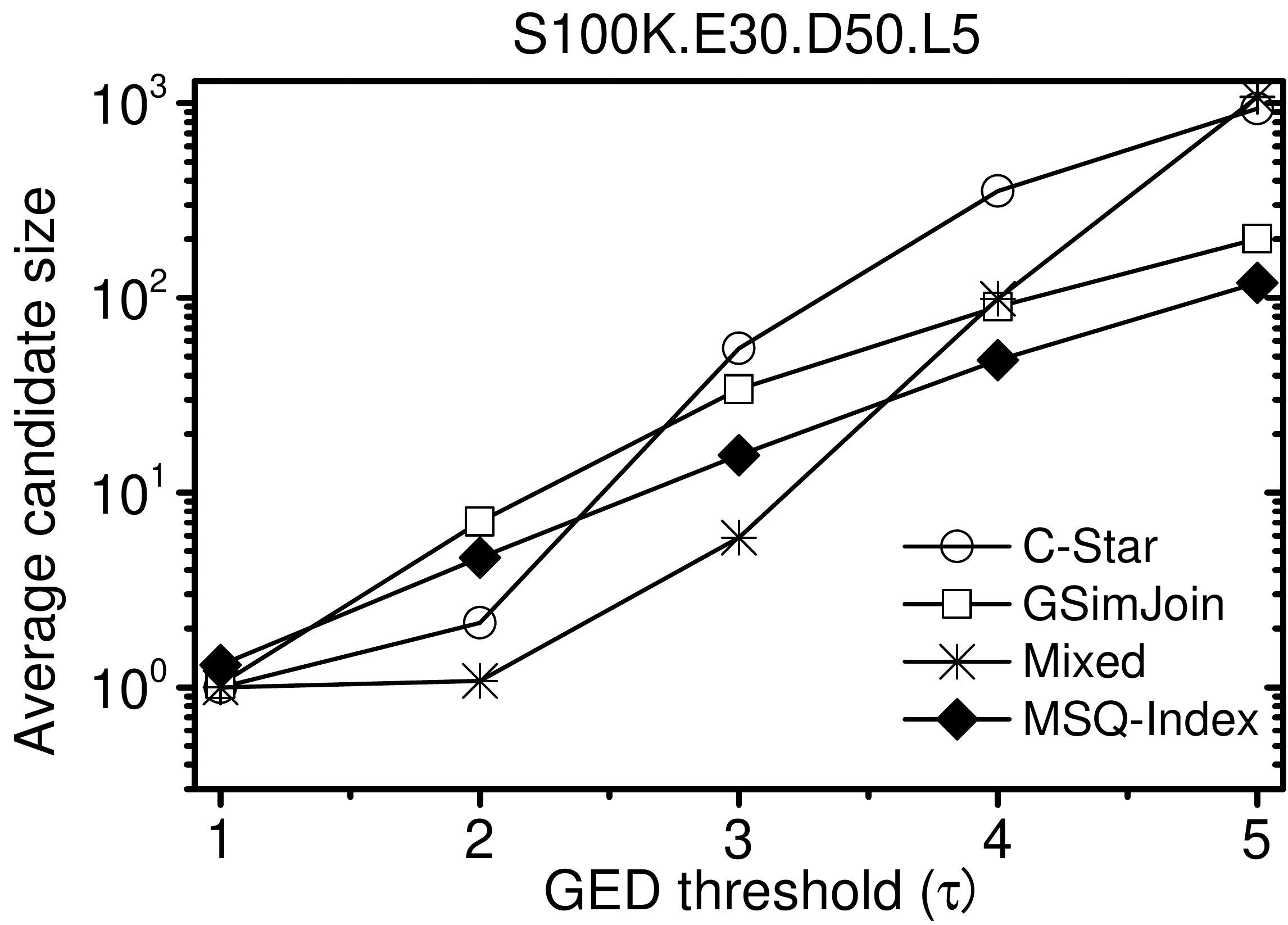}}
	\end{minipage}
	\qquad
	\begin{minipage}{0.28\linewidth}
		\centerline{\includegraphics[width=1\textwidth]{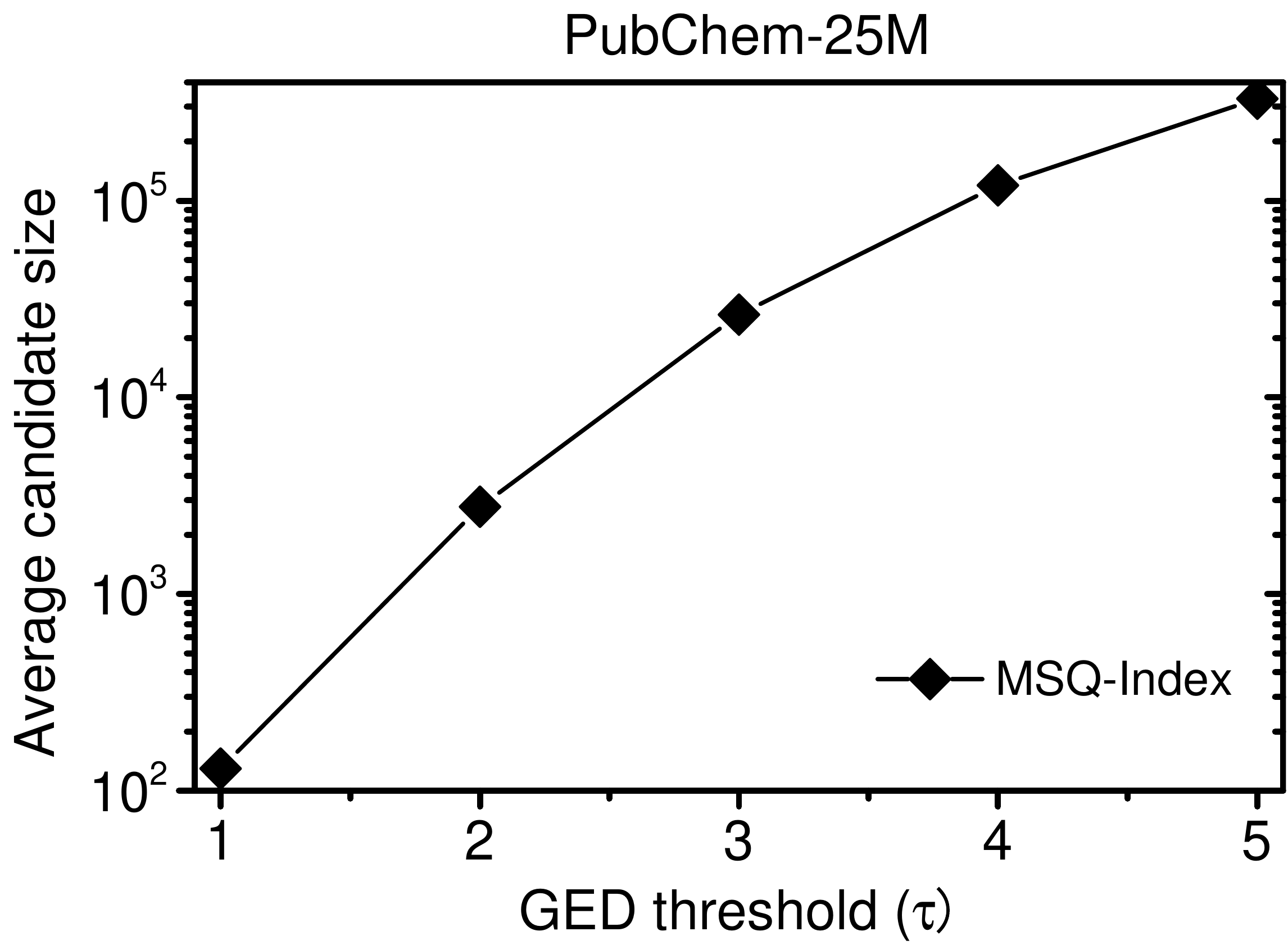}}
    \end{minipage}
    \vfill
    \begin{minipage}{0.28\linewidth}
		\centerline{\includegraphics[width=1\textwidth]{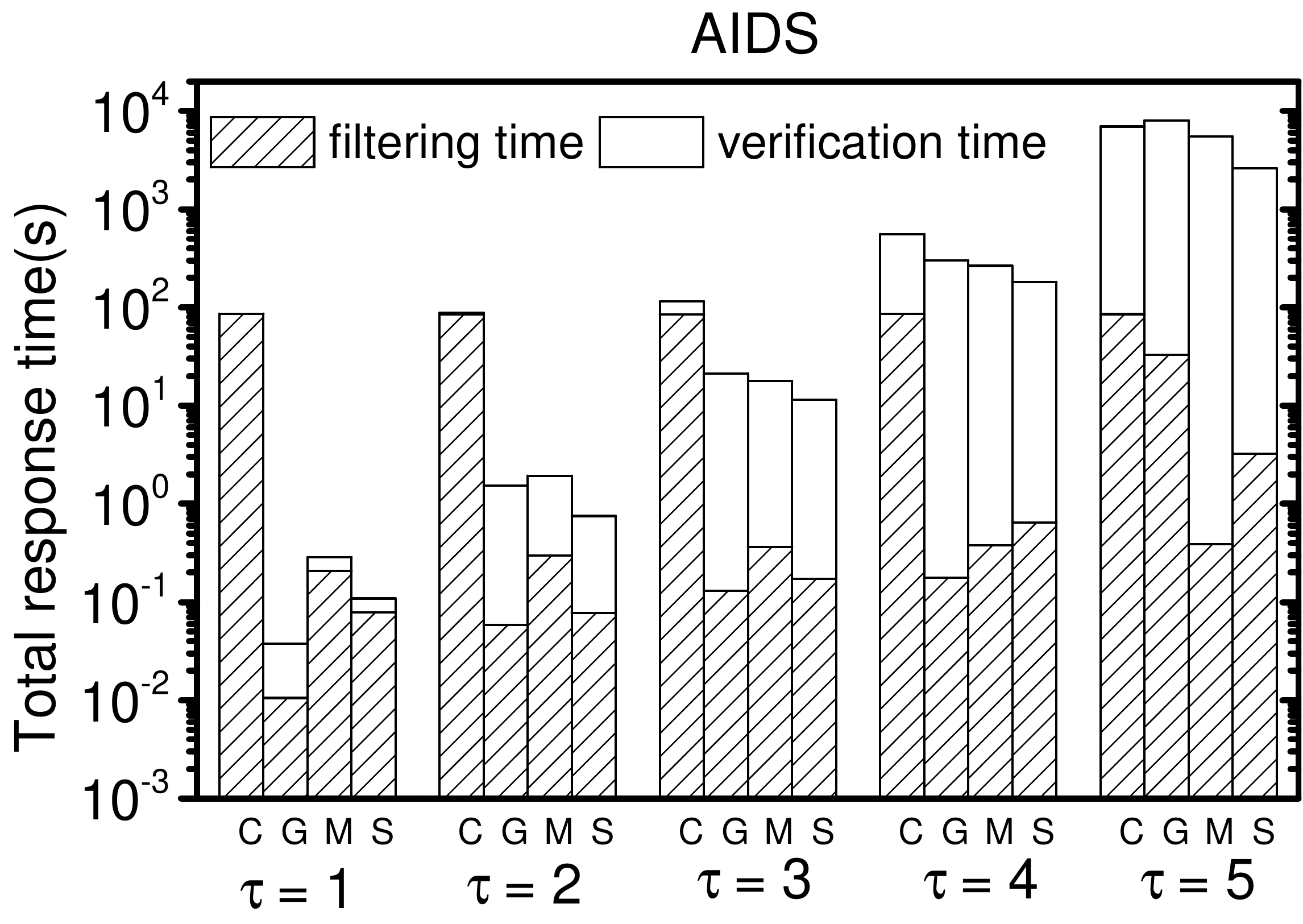}}
	\end{minipage}
	\qquad
	\begin{minipage}{0.28\linewidth}
		\centerline{\includegraphics[width=1\textwidth]{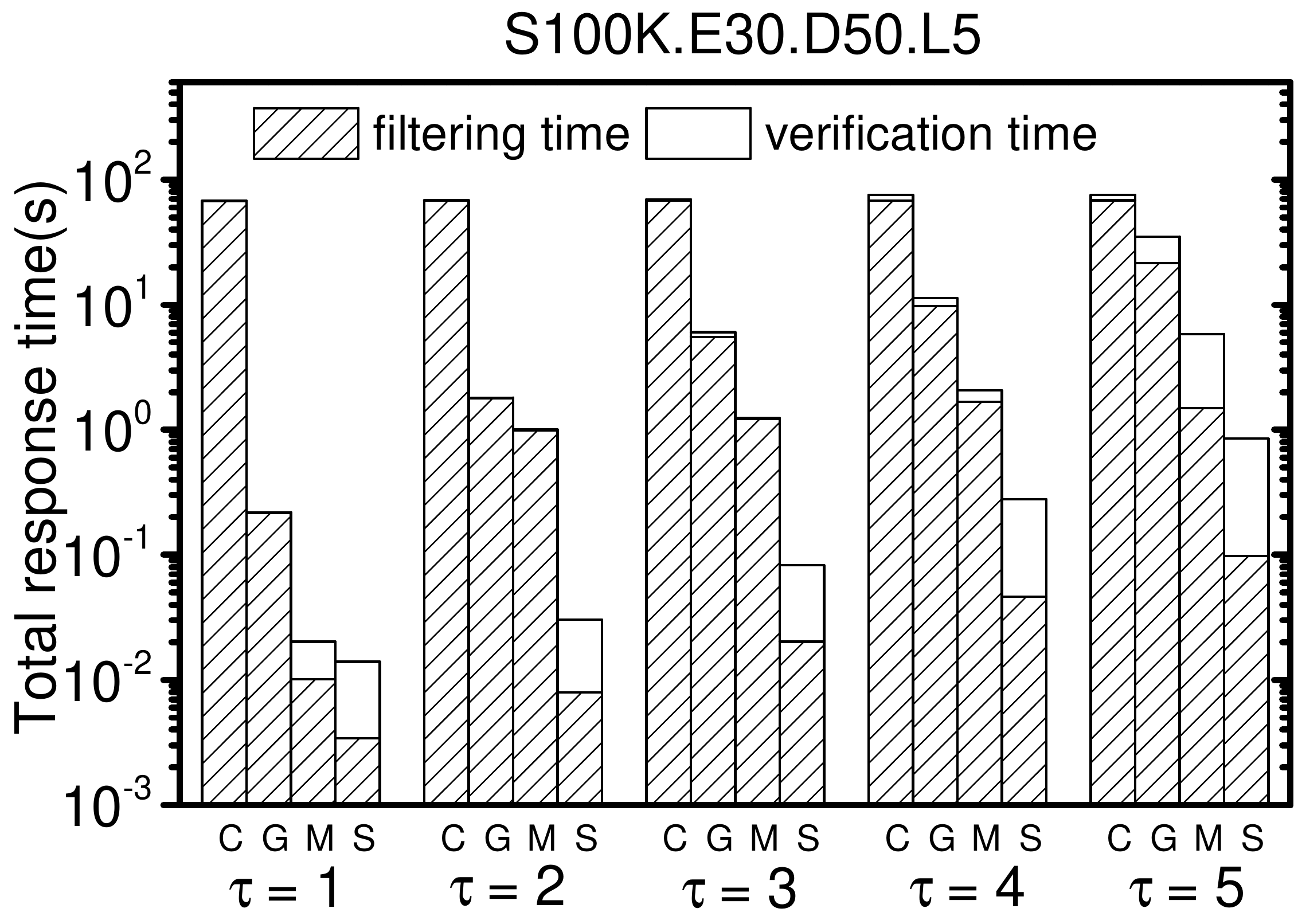}}
	\end{minipage}
	\qquad
	\begin{minipage}{0.28\linewidth}
		\centerline{\includegraphics[width=1\textwidth]{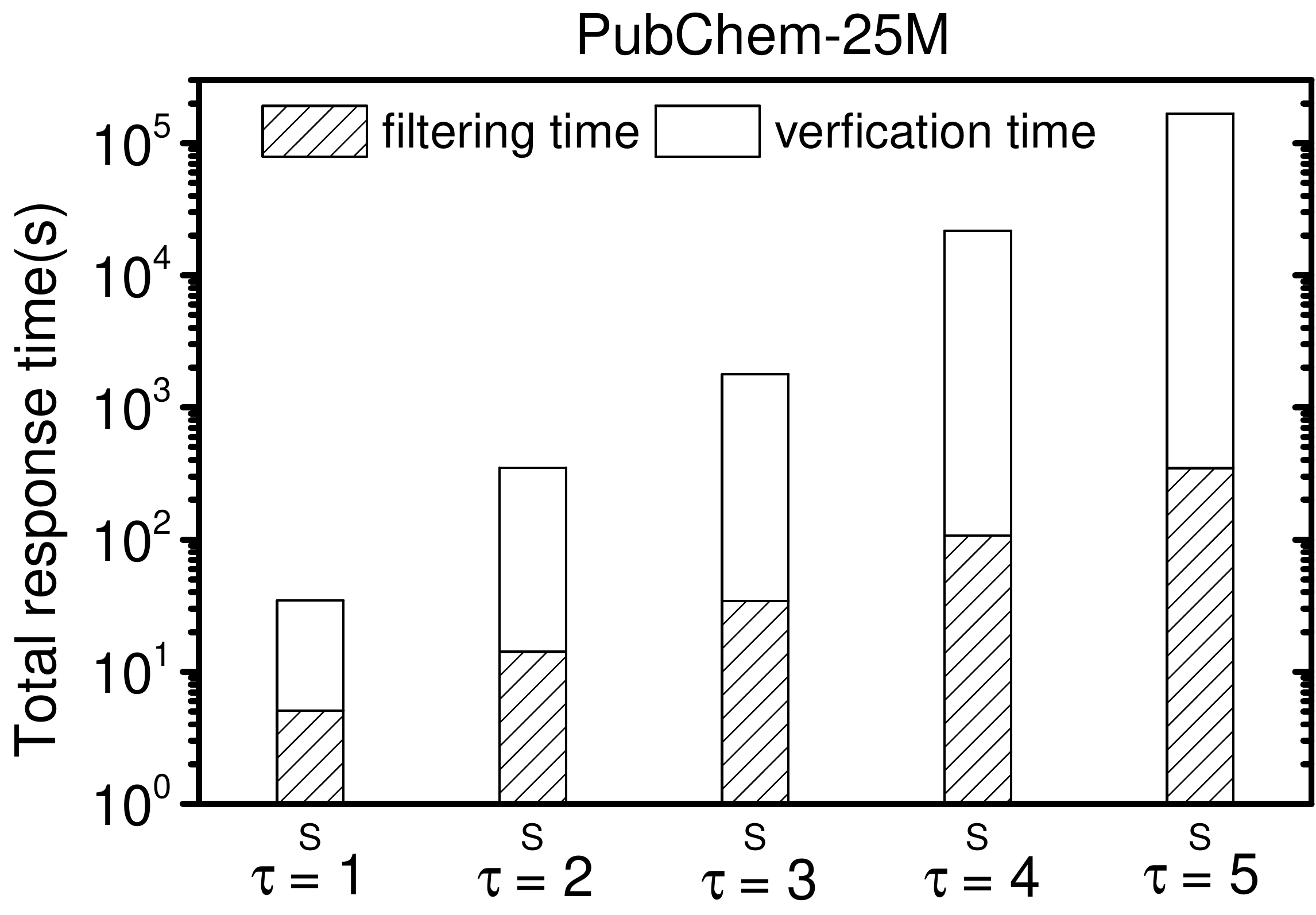}}
	\end{minipage}
    \caption{Average candidate size and total response time on the three datasets.}
    \label{fig:Fig09}
\end{figure*}

\subsubsection{Comparing with Existing Indexes}

We vary the size of datasets to evaluate the index storage space
and construction time, and show the results in Figure~\ref{fig:Fig08}.
Regarding the index size, Mixed consumes the most amount of space
in AIDS and PubChem-25M since it has to store all branch and disjoint
structures. However, GSimJoin does not perform well in S100K.E30.D50.L5
since the number of paths increases exponentially in the dense
graphs. MSQ-Index performs the best and its index size is only
5\% of that of Mixed and 15\% of that of C-Star. This is because
that (1) the total number of degree structures and label structures
is less than the number of tree structures and paths; (2) the
entries in the succinct $q$-$\id{gram}$ tree are compressed for
efficient storage. For the large dataset PubChem-25M, Mixed,
GSimJoin and C-Star cannot properly run for the memory error when
the dataset size is more than 15M, while the index size of MSQ-Index
is about 1.2GB, which achieves an excellent performance.

By Figure~\ref{fig:Fig08}, we know that C-Star has the shortest
index construction time. This is because that it only needs to
enumerate all star structures in each data graph without any complex
index. Although MSQ-Index is stored in a succinct form, its construction
time is shorter than GSimJoin and Mixed. For the large dataset
PubChem-25M, it can be built done in 1 hour.

\subsection{Filter Performance}

In this section, we evaluate the query performance of all tested
methods on the datasets AIDS and S100K.E30.D50.L5 and PubChem-25M
on two metrics: \# of candidates passed the filtering and overall
processing time. For overall processing, we further divide it into
two parts: indexing processing time and candidate verification time.

We fix the datasets and vary the edit distance threshold $\tau$
from 1 to 5 to evaluate the filter capability and response time.
Figure~\ref{fig:Fig09} shows the the average candidate size and
total response time (i.e., the filtering time plus the verification
time) of different methods for the fifty query graphs. Note that,
we combine the heuristic estimate function $h(x)$ in~\cite{ZhaoXLW2012}
into the software provided by Riesen et al.~\cite{RiesenEB2013}
to compute the exact graph edit distance in the verifcation phase
for C-Star, Mixed and MSQ-Index, except for GSimJoin has implemented
it in their executable binary file.

Regarding the candidate size, we can know that our method has the
smallest candidate size in most case. GSimJoin and C-Star do not
perform well because that both tree structures and paths have much
more overlapping. In S100K.E30.D50.L5, Mixed performs the best
when~$\tau \leq 3$ and our method has a close candidate size with
it. For the large dataset PubChem-25M, only our method can properly
run since only it can be built done in our environment.

For the response time of C-Star (denoted by ``C''), GSimJoin
(denoted by ``G''), Mixed (denoted by ``M'') and MSQ-Index
(denoted by ``S''), we know that C-Star consumes the longest
filtering time because it needs to construct a minimum weighted
bipartite graph between each data graph and the query graph. Even
though GSimJoin shows a better filtering time in AIDS, it produces
a large candidate set than MSQ-Index, making the total response
time large than MSQ-Index. Compared with Mixed, MSQ-Index can
achieve 1.6x speedup on AIDS and 3.8x speedup on S100K.E30.D50.L5
on the average. In addition, although MSQ-Index are compressed
for efficient storage, it can provide good filtering efficiency
especially when $\tau$ is small, such as the total filtering time
of MSQ-Index is less than 5s in PubChem-25M when $\tau = 1$.

\subsection{Scalability}

In this section, we evaluate the scalability performance of C-Star,
GSimJoin, Mixed and MSQ-Index on the real and synthetic datasets.

\subsubsection{Varying $|V_h|$}

We vary the query graph size from 10 to 60, and fix the size of
PubChem-25M be 5M and $\tau = 3$, respectively, to evaluate the
effect of the query graph size on the query performance. Figure~\ref{fig:Fig11}
shows the distribution of graphs in $G$, where x-axis is the graph
size, i.e., the number of vertices and y-axis is the number of graphs
of the same size in the dataset. Figure~\ref{fig:Fig12} shows the
average candidate size and total filtering time for the fifty
query graphs, respectively.

\begin{figure}[ht]
\centering
	\begin{minipage}{0.48\linewidth}
		\centerline{\includegraphics[width=1\textwidth]{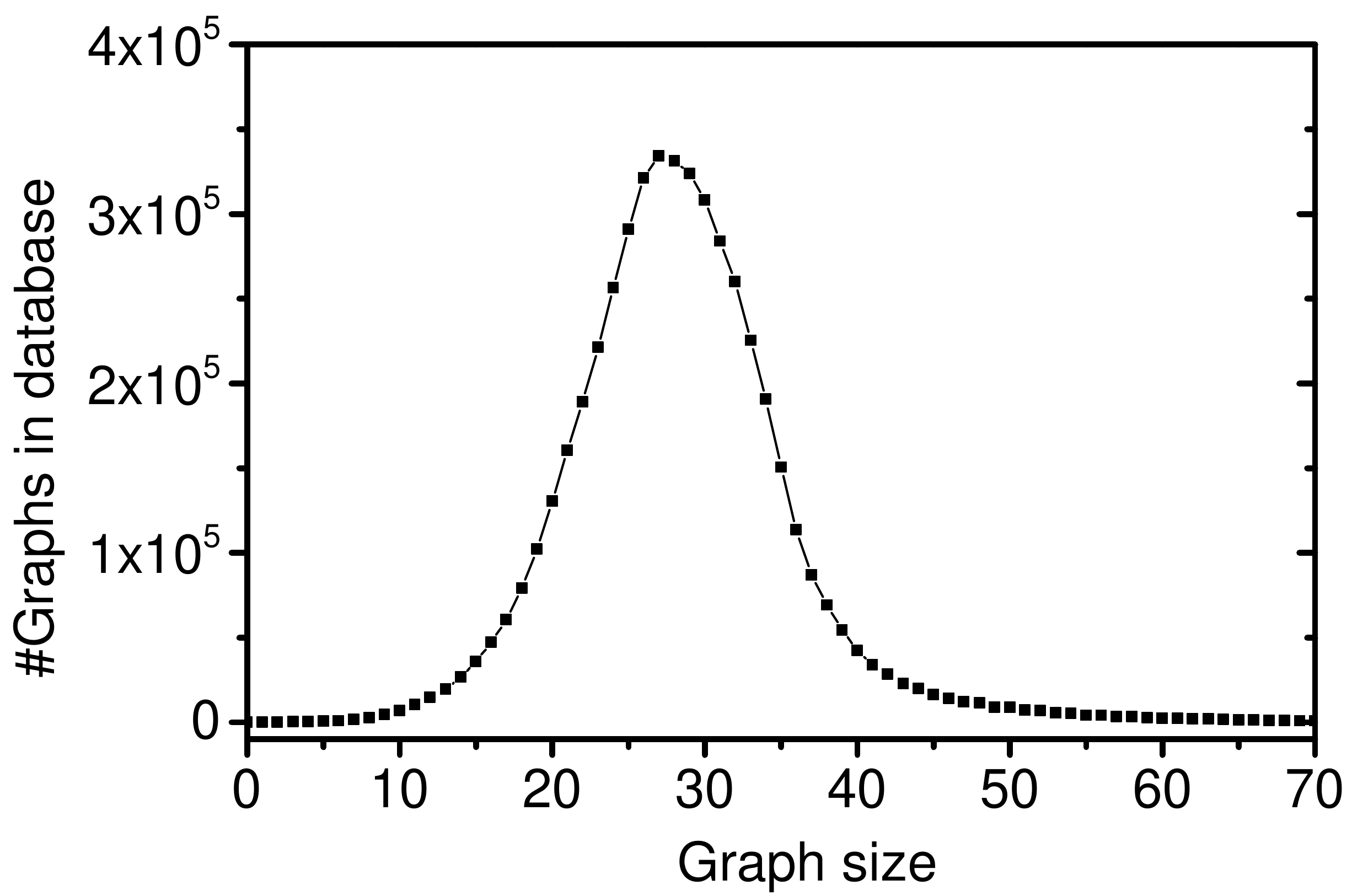}}
	\end{minipage}
    \caption{Graph distribution.}
    \label{fig:Fig11}
\end{figure}

By Figure~\ref{fig:Fig11}, we know that the distribution of
data graphs in $G$ is close to a normal distribution and the number
of graphs whose size near~30 is relatively large. Thus, the average
candidate size of all tested (excepting GSimJoin for the memory error)
methods first increase and then decrease, and achieves the maximum
when the query graph size is~30.

By Figure~\ref{fig:Fig12}(b), we know that MSQ-Index has the shortest
filtering time. Compared with Mixed, MSQ-Index can achieve 8--40x
speedup when the query graph size is less than 20 or more than 50.
The reason is that the number of data graphs whose size near 20 or
50 in the dataset is relatively small by Figure~\ref{fig:Fig11},
resulting in the query region $Q_h$ containing few data graphs.

\begin{figure}[htbp]
\centering
	\begin{minipage}{0.47\linewidth}
		\centerline{\includegraphics[width=1\textwidth]{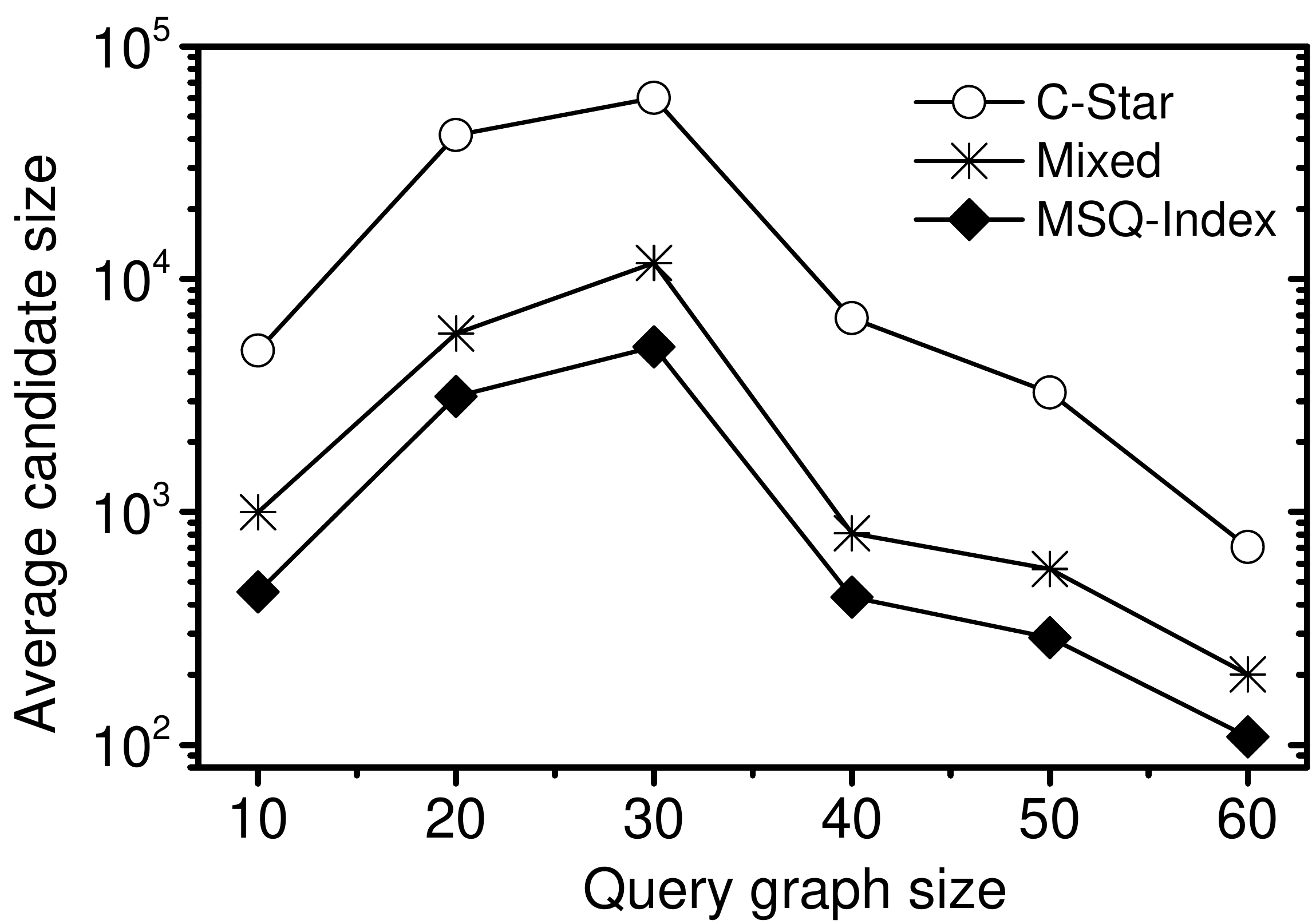}}
		\centerline{(a) Average candidate size}
	\end{minipage}
	\quad
	\begin{minipage}{0.47\linewidth}
		\centerline{\includegraphics[width=1\textwidth]{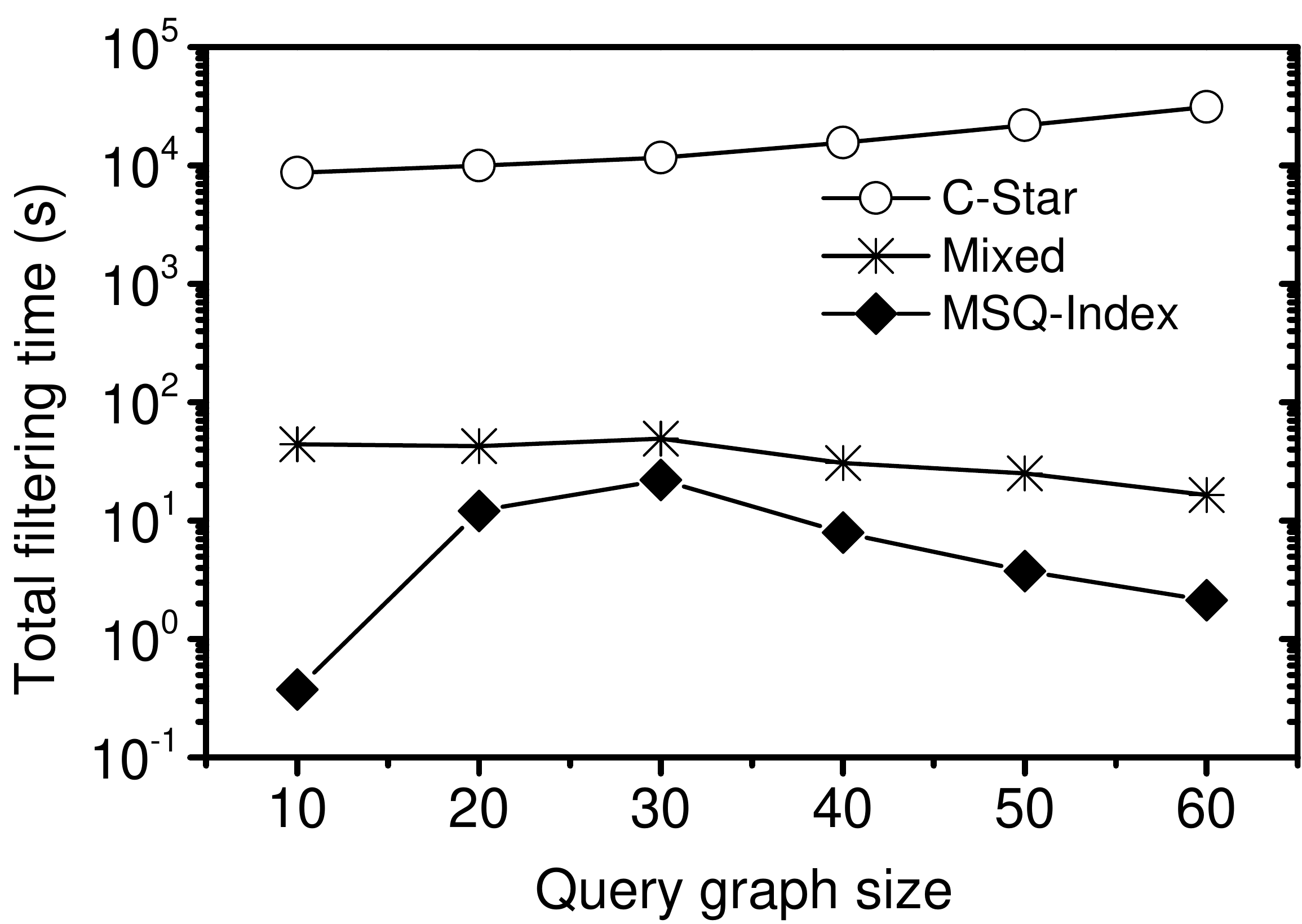}}
		\centerline{(b) Total filtering time}
	\end{minipage}
    \caption{Scalability $vs. \ |V_h|$.}
    \label{fig:Fig12}
\end{figure}

\subsubsection{Varying $|G|$}

We fix $\tau = 5$ and vary the size of PubChem-25M from 500K (kilo)
to 25M (million) to evaluate the effect of the dataset size.
Figure~\ref{fig:Fig10} shows the average candidate size and the
total response time for the fifty graphs. Among all tested methods,
MSQ-Index has the smallest candidate size and the shortest response
time. When the dataset size is 10M, GSimJoin and Mixed cannot
properly run for the memory error, and the verification time of
C-Star is longer than 48 hours, making all of them not be suitable
for such large dataset. Only MSQ-Index can easily scale to cope
with it.

\begin{figure}[htbp]
\centering
    \begin{minipage}{0.47\linewidth}
		\centerline{\includegraphics[width=1\textwidth]{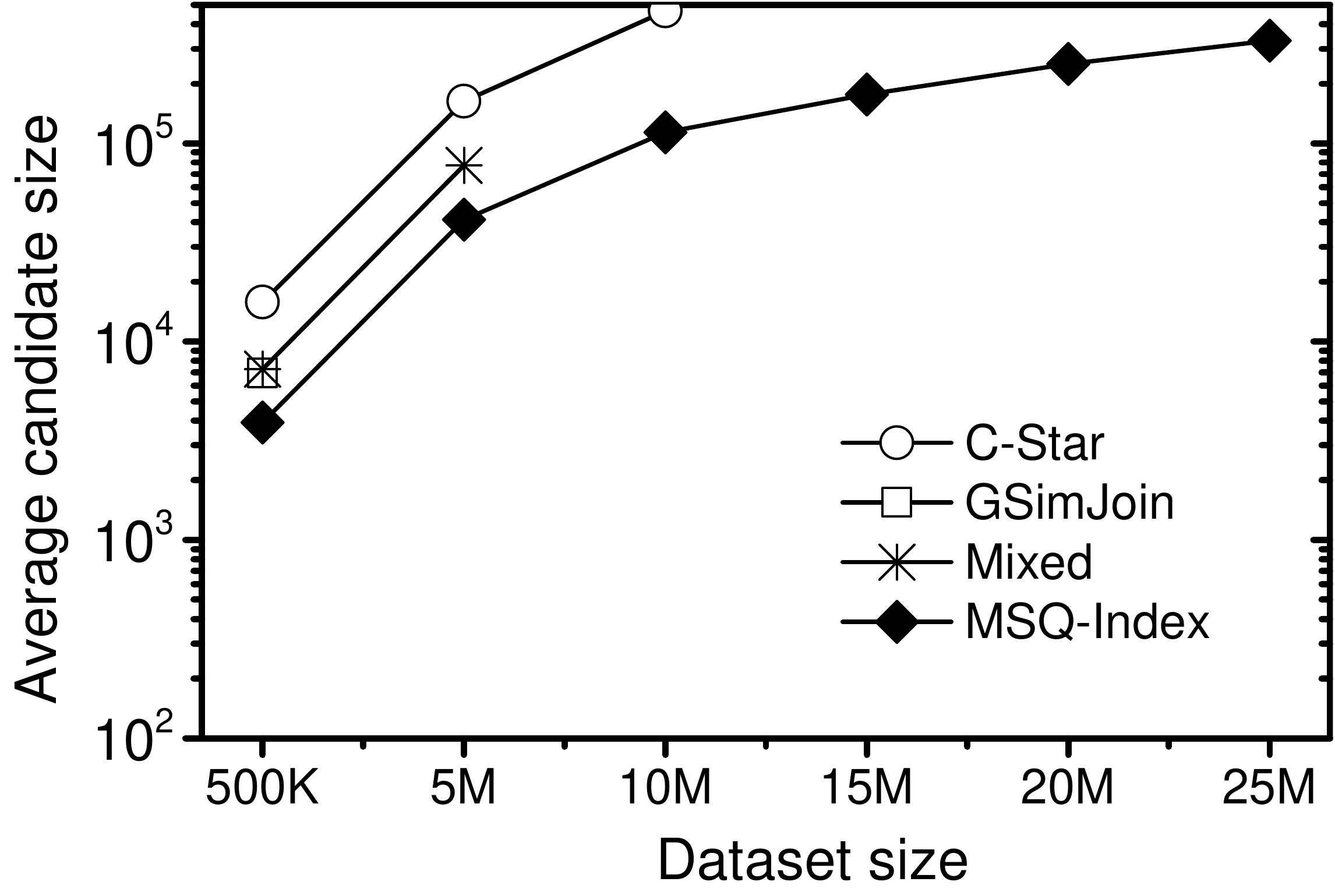}}
		\centerline{(a) Average candidate size}
	\end{minipage}
	\quad
	\begin{minipage}{0.47\linewidth}
		\centerline{\includegraphics[width=1\textwidth]{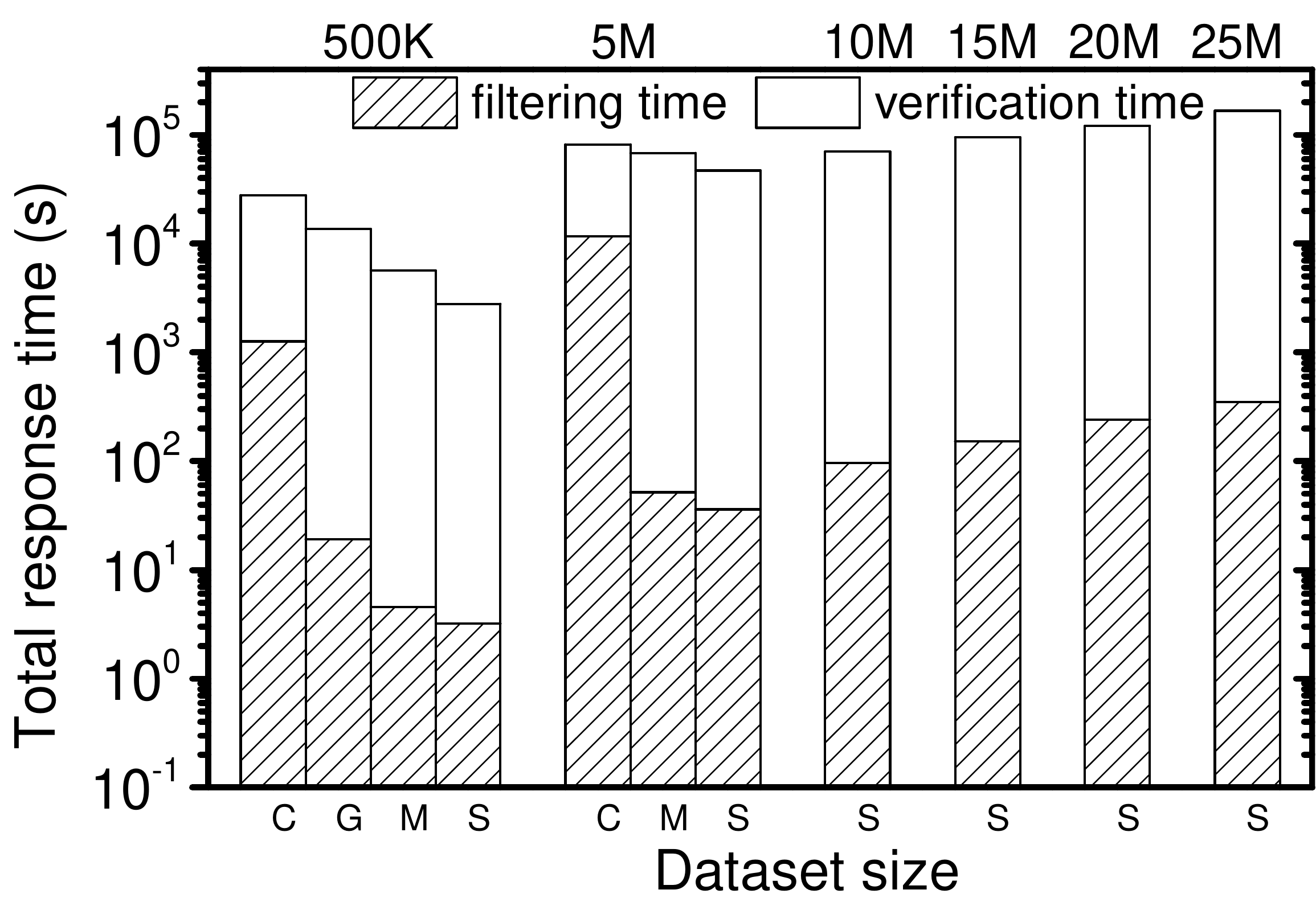}}
		\centerline{(b) Total response time}
	\end{minipage}
	\caption{Scalability $vs. \ |G|$.}
    \label{fig:Fig10}
\end{figure}

\subsubsection{Varying $|\Sigma_{V}|$}

We fix $\tau = 5$, the dataset size be 100K, the average graph
density $\rho= 50\%$, the number of edges in each data graph be 30,
respectively, and then produce a group of synthetic datasets to
evaluate the effect of the number of labels. Figure~\ref{fig:Fig13}
shows the average candidate size of all tested methods. By
Figure~\ref{fig:Fig13}, we know that the average candidate size
decreases as the number of vertex labels increases. This is
because that more information can be used to filter.

\begin{figure}[ht]
\centering
    \begin{minipage}{0.47\linewidth}
		\includegraphics[width=1\textwidth]{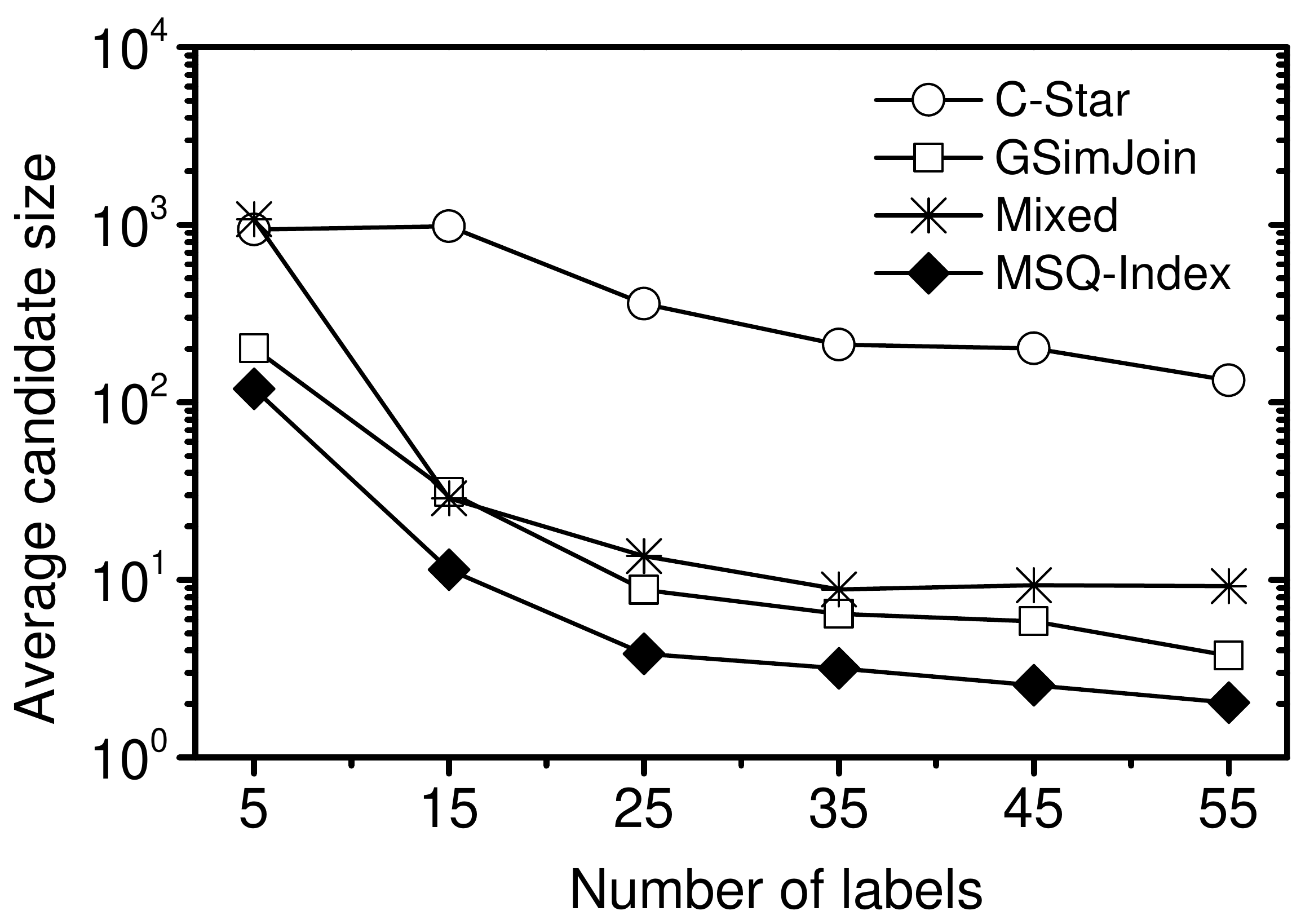}		
	    \caption{\scriptsize{Scalability $vs. \ |\Sigma_V| $}}
        \label{fig:Fig13}
    \end{minipage}
    \quad
    \begin{minipage}{0.47\linewidth}
		\includegraphics[width=1\textwidth]{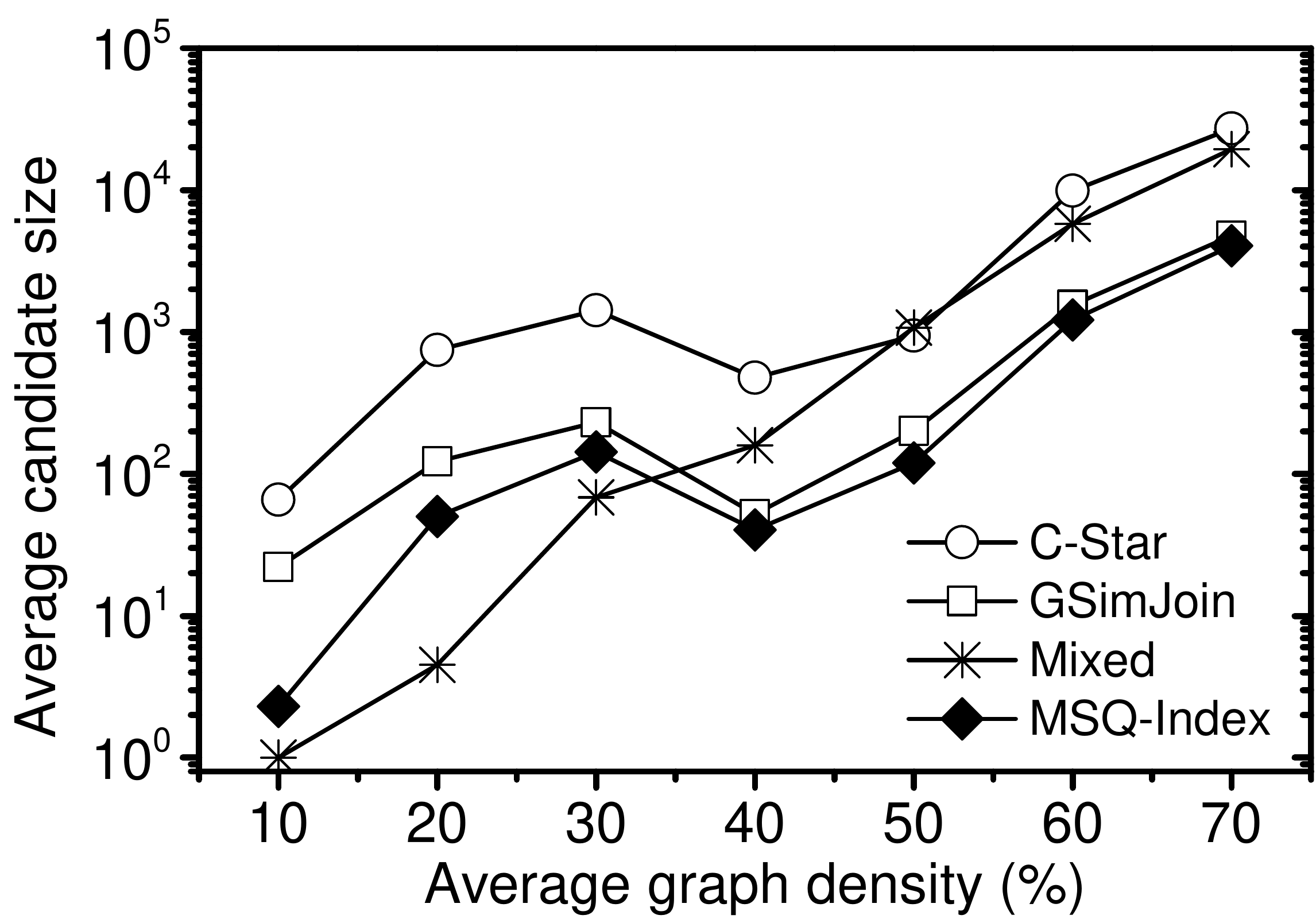}
        \caption{\scriptsize{Scalability $vs. \ \rho$.}}
        \label{fig:Fig14}	
	\end{minipage}
\end{figure}

\subsubsection{Varying $\rho$}

We fix $ \tau = 5$, the dataset size be 100K, the number of edges
and vertex labels in each data graph be 30 and 5, respectively,
and then produce a group of synthetic datasets to evaluate the
effect of the average density. Figure~\ref{fig:Fig14} shows the
average candidate size of all tested methods. It shows that the
candidate size of all tested methods increases as $\rho$ increases
when $\rho \geq 40\%$. This is because that all tested methods
only using the local structures have a weak filter ability for
the the density graphs.

\section{Related Works}
\label{sec:relatedWorks}

Recently, graph similarity search has received considerable
attention. $\kappa$-AT~\cite{WangWYY2012} and GSimJoin~\cite{ZhaoXLW2012}
are two major $q$-$\id{gram}$ counting filters. In $\kappa$-AT,
a $q$-$\id{gram}$ is defined as a tree consisting of a vertex $v$
and the paths whose length no longer than $\kappa$ starting from $v$.
However, GSimJoin considered the simple path whose length is $p$
as a $q$-$\id{gram}$. The principle of the $q$-$\id{gram}$ counting
filter is stated as follows: if $ged(g, h) \leq \tau$, graphs $g$
and $h$ must share at least $\idrm{max}\{|Q(g)|-\gamma_g \cdot \tau, |Q(h)|-\gamma_h \cdot \tau\}$
common $q$-$\id{grams}$, where $Q(g)$ and $Q(h)$ denote the multisets
of $q$-$\id{grams}$ in $g$ and $h$, respectively, $\gamma_g$ and $\gamma_h$
are the maximum number of $q$-$\id{grams}$ that can be affected by
an edit operation, respectively. C-Star~\cite{ZengTWFZ2009} and
Mixed~\cite{ZhengZLWZ2013,ZhengZLWZ2015} are two mapping distance-based
$\idrm{filters}$. The lower bounds are $L_S(g,h) = \frac{s_m (g,h)}{max\{4,max\{d_g,d_h\}+1\}}$ and $L_B(g, h) = \frac{b_m (g,h)}{2}$,
respectively, where $s_m(g, h)$ and $b_m(g, h)$ are the mapping
distances derived based on the minimum weighted bipartite graphs
between the star and branch structures of $g$ and $h$, respectively,
and $d_g$ and $d_h$ are the respective maximum degrees in $g$
and $h$. SEGOS~\cite{WangDTYJ2012} introduced a two-level index
structure to speed up the filtering process, which has the same
filter ability with C-star. Pars~\cite{ZhaoXLLZ2013} divided each
data graph $g$ into $\tau + 1$ non-overlapping substructures, and
pruned the graph $g$ if there exists no substructure that is subgraph
isomorphic to $h$. The above methods show different performance on
different databases and we can hardly prove the merits of them
theoretically~\cite{GoudaA2015}.

In the verification phase, $A^*$ algorithm~\cite{HartNB1968} is widely
used to compute the exact graph edit distance. Zhao et.al~\cite{ZhaoXLW2012}
and Gouda et.al~\cite{GoudaH2016} designed different heuristic estimate
functions to improve $A^*$. Note that, we only focus on the filtering
phase in this paper.

When the database contains millions of graphs, many existing approaches
cannot properly run. gWT~\cite{TabeiT2011} utilized the Weisfeiler-Lehman (WL)
kernel~\cite{ShervashidzeB2009} function to compute the similarity
between two graphs and constructed a wavelet tree~\cite{GrossiGV2003}
to speed up the query processing. Chen et.al~\cite{ChenXBXC2014} built
the index structure on a hadoop~\cite{DeanG2008} cluster to search on
the large database. Unlike previous methods, we propose the first
succinct index structure for this problem for efficient storage.
Our index can scale to cope with the large dataset of millions of graphs.

\section{Conclusions and Future Work}
\label{sec:conclusion}

We present an space-efficient index structure for the graph similarity
search problem, whose encoded sequence~$\id{S_X}$ requires
$|\id{\varPsi_X}|\log b_X^{m} + |\id{\varPsi_X}|$ bits, where~$X$
denotes~$D$ or $L$. Our index structure incorporates succinct data
structures and hybrid encoding to significantly reduce the index
space usage while at the same time keeping fast query performance.
Each entry in $\id{\varPsi_X}$ requires about between 4 and 6 bits
on our data, which is much smaller than that used to represent an
entry in the compared indexing methods in this paper. However,
there is still room for improvement on this space bound of $\id{S_X}$.
The design of a representation of the $q$-$\id{gram}$ tree that
achieves the entropy-compressed space bound while still preserving
query efficiency is left as a future work.

\balance

\section{Acknowledgments}
The authors would like to thank Weiguo Zheng and Lei Zhou for providing
their source files, and thank Xiang Zhao and Xuemin Lin for providing
their executable files.This work is supported in part by China NSF
grants 61173025 and 61373044, and US NSF grant CCF-1017623. Hongwei Huo
is the corresponding author.


\end{document}